\documentclass[12pt]{article}

\usepackage[T1]{fontenc}
\usepackage{lmodern}
\usepackage{amssymb}
\usepackage{amsfonts}
\usepackage{amsmath}
\usepackage[nohead]{geometry}
\usepackage[singlespacing]{setspace}
\usepackage[bottom]{footmisc}
\usepackage{indentfirst}
\usepackage{endnotes}
\usepackage[pdftex,xdvi]{graphicx}
\usepackage{epstopdf}
\usepackage{rotating}
\usepackage{verbatim}
\usepackage{setspace}
\usepackage{multirow}
\usepackage{latexsym}
\usepackage{tikz}
\usetikzlibrary{arrows.meta,
chains,
positioning,
shapes.geometric
}
\usepackage{color}
\usepackage{amsthm}
\usepackage{makeidx}
\usepackage{fancyhdr}
\usepackage{type1cm}
\usepackage{mathtools}
\usepackage{blindtext}
\usepackage{tocloft}

\usepackage{dcolumn}
\usepackage{pdfpages}

\usepackage{caption}
\usepackage{booktabs}
\usepackage{longtable}
\usepackage{array}
\usepackage{wrapfig}
\usepackage{float}
\usepackage{colortbl}
\usepackage{pdflscape}
\usepackage{threeparttable}
\usepackage{threeparttablex}
\usepackage[normalem]{ulem}
\usepackage{makecell}
\usepackage{xcolor}
\usepackage{subcaption}
\usepackage{morefloats}
\usepackage{placeins}
\usepackage{enumerate}
\usepackage{algorithm}
\usepackage{algpseudocode}
\usepackage{ragged2e}

\usepackage{textcomp}
\usepackage[hidelinks,hypertexnames=false]{hyperref}
\hypersetup{
pdftitle={Minimax Choice of Projection Geometry under Linear Inequality Constraints},
pdfauthor={Joachim Freyberger and Julius Kappenberg},
pdfkeywords={shape restrictions, inequality constraints, projection geometry, minimax risk}
}
\usepackage[round]{natbib}

\newcommand{\R}{\ensuremath{\mathbb{R}}}

\newcommand{\1}{\ensuremath{\mathbf{1}}}

\newtheoremstyle{dotless}{}{}{\itshape}{}{\bfseries}{}{ }{}

\theoremstyle{dotless}

\newtheorem{theorem}{Theorem}[section]
\newtheorem{corollary}{Corollary}[theorem]
\newtheorem{lemma}[theorem]{Lemma}
\newtheorem{proposition}[theorem]{Proposition}

\newtheorem{example}{Example}

\DeclareMathOperator*{\argmin}{arg\,min}

\long\def\symbolfootnote[#1]#2{\begingroup%
	\def\thefootnote{\fnsymbol{footnote}}\footnote[#1]{#2}\endgroup}

\begin{document}
	
\allowdisplaybreaks

\setstretch{1.5}

\begin{center}
	
	\quad
	\vspace{0mm}

\Large{\bf{Minimax Choice of Projection Geometry under Linear Inequality Constraints}}\symbolfootnote[1]{We thank Antonia Antweiler, Tim Armstrong, Stéphane Bonhomme, Bj\"orn H\"oppner, Kirill Ponomarev, Jan Scherer, and audiences at the IAAE Annual Conference 2026 and the 2026 European Meeting of the Econometric Society for very helpful comments and discussions. The research was supported by the European Research Council (ERC-2020-STG-949319).}

\vspace{7mm}

\normalsize 

Joachim Freyberger\symbolfootnote[2]{University of Bonn. Email: freyberger@uni-bonn.de} \hspace{20mm}Julius Kappenberg\symbolfootnote[3]{University of Bonn. Email: julius.kappenberg@uni-bonn.de}

\vspace{7mm}
	

\today
	
\vspace{1cm}

\end{center}

\normalsize

\setstretch{1.25}

\noindent \textbf{Abstract} \medskip

\noindent Economic theory frequently implies linear inequality restrictions on parameters or functions of interest. A common way to impose such restrictions is to project an unrestricted estimator onto the feasible set. Projection estimators arise naturally from constrained least squares, instrumental variables, generalized method of moments, maximum likelihood, and related extremum procedures. When the sampling covariance, loss function, and projection criterion induce different geometries, the choice of projection geometry can substantially affect risk. We study this choice in a fixed-dimensional local Gaussian experiment under quadratic loss. At exact-boundary configurations where only one maintained inequality binds, inverse-covariance projection is pointwise optimal. When at most two inequalities are locally relevant, it weakly improves on the unrestricted estimator throughout the corresponding local experiment and is minimax over exact-boundary configurations. For an arbitrary number of inequalities, we provide a sufficient condition for boundary minimaxity, but show by counterexample that inverse-covariance projection need not be boundary minimax once three inequalities can bind. Motivated by these results, we propose selecting the projection geometry to minimize worst-case boundary risk subject to a local no-harm condition relative to unrestricted estimation. We develop a feasible implementation and study its finite-sample performance in simulations and an application to gasoline demand.

\vspace{1cm}

\newpage

\setstretch{1.5}

\section{Introduction} \label{sec:introduction}
Economic theory frequently implies inequality restrictions on objects of interest. Demand functions are typically required to be decreasing in prices, production functions may satisfy monotonicity and concavity restrictions, treatment effects may be ordered across treatment intensities, estimated quantile curves should not cross, and probability weights are subject to non-negativity restrictions, often together with adding-up equalities. Such restrictions can contain useful information and improve estimation precision.

A common way to incorporate these restrictions is through projection. Starting from an unrestricted estimator $\hat\theta_{ur}$, the researcher projects it onto the feasible parameter space whenever the unrestricted estimate violates the maintained inequalities. Projection-based estimators arise naturally from constrained versions of standard extremum procedures, including least squares, instrumental variables (IV), generalized method of moments (GMM), and maximum likelihood. They also have an appealing stability property: if the true parameter lies in the strict interior of the feasible set, the restricted and unrestricted estimators coincide with probability approaching one. Hence, imposing the restrictions does not alter the first-order behavior of the estimator when the restrictions are asymptotically irrelevant.

Projection estimators therefore form a natural and practically important class of constrained estimators. Within this class, the norm used for projection is an important choice. Different projection geometries impose the same restrictions but can move the same unrestricted estimate in different directions, generating substantially different restricted estimates and risks when the constraints bind. Given an unrestricted estimator and a set of maintained linear inequality constraints, we ask which projection geometry delivers desirable risk properties and how the answer depends on the sampling covariance, the loss function, and the geometry of the constraints.

To make the distinction explicit, suppose the unrestricted estimator is asymptotically normal with covariance matrix $\Sigma$, and let $H$ denote the quadratic loss used to evaluate estimation error. A projection estimator additionally requires a positive-definite matrix $\Omega$ that determines how an infeasible unrestricted estimate is moved toward the constraint set. Thus there are three potentially distinct objects: the sampling geometry $\Sigma$, the loss geometry $H$, and the projection geometry $\Omega$. Their interaction is central to the paper. We use the term \emph{projection direction} when the local problem reduces to one locally relevant inequality and \emph{projection geometry} for the general problem with several constraints. As illustrated in Section~\ref{subsec:examples}, the projection geometry may be inherited from the estimation criterion, as in constrained maximum likelihood or GMM, or chosen as a separate post-estimation step. For exact linear equality restrictions, inverse-covariance projection is classically optimal within this class \citep{Theil1971}. Under inequalities, the relevant face depends on the realization of the unrestricted estimator, so the geometry choice becomes more subtle.

We first consider one locally relevant inequality. In the Gaussian local experiment, inverse-covariance projection is optimal at the exact boundary and weakly improves on unrestricted estimation throughout the corresponding local experiment. Unless the covariance and loss projection directions coincide, it is neither pointwise optimal at finite nonzero local slack nor minimax over any local neighborhood of positive width.

When several inequalities can bind locally, no single projection geometry need be pointwise optimal at every configuration, even on the boundary. We therefore study minimax risk over exact binding configurations, where each inequality is either binding or asymptotically irrelevant. This criterion is tractable and, unlike minimaxity over a bounded local neighborhood, does not require a researcher-chosen neighborhood. Minimaxity over the enlarged local-slack set is not informative for choosing projection geometry because sufficiently interior configurations pin down the minimax value at unrestricted risk. Minimaxity over all boundary configurations with arbitrary finite local slack in nonbinding inequalities would also be natural, but our analytical results do not characterize that problem. We use exact binding configurations as a transparent, tuning-free benchmark.

When at most two inequalities can bind locally, inverse-covariance projection is boundary minimax. For an arbitrary number of inequalities, the same result holds if the worst boundary risk under inverse covariance is attained at a one-binding configuration. This condition holds in many of our numerical examples, but we also construct a three-constraint example in which it fails and another projection geometry has strictly lower worst-case boundary risk. With several constraints, inverse-covariance projection can also have greater risk than unrestricted estimation. Thus neither boundary minimaxity nor the one-inequality guarantee of weak improvement over unrestricted estimation extends automatically.

These results motivate a practical selection problem. We propose choosing the projection geometry to minimize worst-case boundary risk subject to a \emph{local no-harm condition}, which requires that the restricted estimator have no larger asymptotic risk than the unrestricted estimator over the enlarged local-slack set. When inverse covariance satisfies both our sufficient condition for boundary minimaxity and the local no-harm condition, it solves this problem over all projection geometries generated by positive-definite metrics. More generally, direct optimization over the positive-definite cone is difficult, so we develop a feasible implementation along a one-dimensional path between the projection based on the loss geometry and inverse-covariance projection. If a researcher instead wants to use a prespecified projection geometry without implementing these diagnostics, inverse covariance remains a well-motivated default. It has especially strong analytical support for configurations with one or two inequalities that can bind locally and performs well in the economically motivated simulations and application below. The procedure provides a safeguard when interactions among several constraints make that default less attractive.

Simulations illustrate the finite-sample importance of projection geometry. Gaussian designs with interacting constraints show how the proposed selector behaves when inverse covariance can fail. Ordered-treatment simulations study an economically motivated setting in which the diagnostic supports inverse covariance, while a flexible IV design compares constrained criteria that yield different restricted estimators from the same unrestricted estimator. Section~\ref{sec:application} applies the framework to gasoline demand using \citet{blundell2012measuring}.

\textbf{Related literature:} This paper connects to several strands of work on shape restrictions, projection estimators, and minimax decision theory.

\emph{Shape-restricted and order-restricted estimation.}
There is a long literature on estimation under order and shape restrictions. Early contributions include \citet{Hildreth:54}, \citet{Brunk:55}, and \citet{ayer1955empirical}, with systematic treatments in \citet{BarlowEtAl1972} and \citet{RobertsonEtAl1988}. See also \citet{chetverikov2018econometrics} and \citet{johnson2018shape}. Weighted projection is well established \citep{OonoShinozaki2005}. For particular order-restricted models, \citet{Lee1988} studies domination under alternative projection weights, while \citet{GargMisra2025} derive parameter-dependent risk-minimizing weights and admissibility results within isotonic families under componentwise squared-error loss. Neither characterizes the general exact-boundary minimax problem studied here. Closely related to our no-harm comparisons, \citet{RuedaSalvador1995} establish MSE domination of the restricted Gaussian MLE for arbitrary linear functionals with one or two inequalities, while \citet{ShinozakiChang1999} characterize failures under positive-orthant restrictions. Related work studies admissibility, minimaxity, domination, and shrinkage under restricted parameter spaces \citep{SilvapulleSen2005,vanEeden2006,MarchandStrawderman2004,TsukumaKubokawa2008}.

\emph{Risk of projection estimators and tangent-cone geometry.}
A substantial literature studies risk of projection estimators under convex and shape restrictions. \citet{Chatterjee2014}, \citet{CGS:15}, \citet{Bellec2018}, and \citet{GuntuboyinaSen2018} relate risk to local facial or tangent-cone structure and provide sharp adaptive and oracle bounds. In isotonic regression, Euclidean least squares adapts to unknown monotone complexity, with much smaller risk for few constant blocks than in worst-case settings. Related results allow dependence, heteroskedasticity, and separable criteria \citep{Zhang:02}; see also \citet{AmelunxenEtAl2014} and \citet{OymakHassibi2016}.

These results can analyze a specified projection procedure. However, changing the projection geometry also changes the geometry entering the risk problem, so even sharp characterizations for one geometry do not generally rank competing projection geometries under a common loss. We instead allow the sampling covariance $\Sigma$, loss geometry $H$, and projection geometry $\Omega$ to differ and treat $\Omega$ itself as the decision variable. Adapting existing tangent-cone and isotonic risk methods to projection-geometry choice in nonparametric monotone regression would require analyzing that choice in an increasing-dimensional setting. Our results instead use fixed-dimensional arguments for general linear inequalities.

\citet{fang2014optimal} provides a close decision-theoretic connection. He develops a local asymptotic minimax approach to plug-in estimation of directionally differentiable functionals and considers metric projection onto shape-restricted convex sets. The metric used for projection is fixed in his analysis, and the decision concerns the plug-in procedure. In his projection example, the underlying regular parameter may lie outside the restricted set, so varying the projection geometry would generally change the target under misspecification. We instead maintain the inequalities and hold the unrestricted estimator fixed, so alternative projection geometries target the same parameter.

\emph{Minimaxity and adaptation under restrictions.}
Related work studies minimax estimation and inference over restricted parameter spaces \citep{donoho1994,armstrong2018,CaiLow2004Minimax}. In our problem, however, a full-local minimax comparison is not informative for choosing projection geometry because sufficiently interior configurations pin down the minimax value at unrestricted risk. A finite neighborhood can restore discrimination, but our finite-neighborhood result shows that the solution then depends on the researcher-chosen neighborhood. We therefore use exact-boundary configurations only to rank projection geometries, not as prior information that the true parameter lies on the boundary. The estimator continues to impose the full inequality system, while the enlarged local-slack set provides the domain for our no-harm requirement. This separation between a favorable performance target and a broader robustness requirement is related to work on adaptation \citep{CaiLow2004Adaptation,CLX:13,Armstrong2015,KwonKwon2020}. Characterizing minimax procedures over broader estimator classes is outside our scope.

\emph{Constrained asymptotics and econometric shape restrictions.}
Projection is generally nondifferentiable at boundary points. \citet{Wright:81}, \citet{Geyer:94}, and \citet{Andrews:99} derive nonstandard limits for constrained estimators, while \citet{FangSantos2019} study inference for directionally differentiable functionals. Econometric work studies shape-restricted identification, estimation, and inference, including nonparametric IV and related models \citep{chetverikov2017nonparametric,FH:15,HL:17} and settings in which constraints may bind \citep{Dumbgen:03,freybergerreeves2018,chernozhukov2023constrained,cox2024}. These papers study what restrictions imply for estimation or inference; we study the geometry used to impose $A\theta\geq b$.

\emph{Alternative ways of imposing restrictions.}
Projection is not the only way to enforce restrictions. Alternatives include rearrangement and other shape-enforcing operators \citep{CFG:09,chernozhukov2010quantile,chen2021shape}, Bayesian procedures and posterior projection \citep{astfalck2024posteriorprojectioninferenceconstrained}, and procedures that first estimate which constraints bind \citep{HL:17}. These methods use restrictions differently from a metric projection of a fixed unrestricted estimator. We focus on projection estimators because they leave feasible unrestricted estimates unchanged and let the correction geometry be chosen by risk.

Taken together, these literatures motivate projection as an important estimator class without implying its optimality among all estimators. Our contribution is to characterize exact-boundary minimax choice of projection geometry under general linear inequalities and quadratic loss, and to combine this criterion with a local no-harm requirement while allowing sampling covariance, loss, and projection geometries to differ.

\textbf{Structure:} Section~\ref{sec:optimality} introduces the framework and motivating examples. Section~\ref{sec:results} develops the risk criteria and main theoretical results. Section~\ref{sec:selection} presents the practical selection procedure, Section~\ref{sec:simulations} studies its finite-sample performance, and Section~\ref{sec:application} contains the gasoline-demand application. Section~\ref{sec:summary} concludes. The appendices contain proofs, computational details, and additional simulations.

\section{Projection geometry under inequality constraints} \label{sec:optimality}
\subsection{General setup and notation}
\label{subsec:setup}

Let $\theta_0\in\mathbb{R}^k$ denote the finite-dimensional parameter of interest and suppose that it satisfies the linear inequalities
\[
\Theta_R=\{\theta\in\mathbb{R}^k:A\theta\geq b\},
\]
where $A\in\mathbb{R}^{d\times k}$ and $b\in\mathbb{R}^d$. We assume throughout that $\Theta_R$ is nonempty, has nonempty interior in $\mathbb R^k$, that the maintained representation does not contain redundant constraints, and that $\theta_0\in\Theta_R$. Rank conditions are imposed locally whenever a particular active-row inverse is used. Maintained equality restrictions are absorbed before defining $\theta$, so the inequalities are written in reduced coordinates.

Let $\hat\theta_{ur}$ be an unrestricted estimator satisfying
\[
\sqrt n(\hat\theta_{ur}-\theta_0)\xrightarrow{d}\mathcal V_{ur},
\qquad
\mathcal V_{ur}\sim N(0,\Sigma),
\]
where $\Sigma$ is positive definite, denoted $\Sigma\succ0$. For a positive-definite matrix $\Omega$, define $\|x\|_\Omega^2=x'\Omega x$. We study projection estimators
\[
\hat\theta_r(\hat\Omega)
=
\argmin_{\theta\in\Theta_R}\|\theta-\hat\theta_{ur}\|_{\hat\Omega}^2,
\qquad
\hat\Omega\xrightarrow{p}\Omega,
\]
as well as constrained extremum estimators that are asymptotically equivalent to such projections.
As shown in Appendix~\ref{appendix:proj_details}, restricting attention to symmetric positive-definite projection metrics is without loss of generality. Multiplying $\Omega$ by a positive scalar does not change the projection, so only the geometry induced by $\Omega$ matters. When the local problem reduces to a single relevant inequality with normal vector $a$, $\Omega$ affects the estimator only through $d_\Omega=\Omega^{-1}a/(a'\Omega^{-1}a)$. We therefore use the term \emph{projection direction} in the single-halfspace analysis and \emph{projection geometry} more generally.

We evaluate estimators under quadratic loss
\[
L_H(\hat\theta,\theta_0)
=(\hat\theta-\theta_0)'H(\hat\theta-\theta_0),
\qquad H\succ0.
\]
It is useful to keep $H$, $\Sigma$, and $\Omega$ conceptually distinct. The matrix $\Sigma$ describes sampling uncertainty, $H$ describes the loss used to evaluate estimation error, and $\Omega$ describes the geometry used to impose the restrictions. Appendix~\ref{appendix:proj_details} shows that the theoretical analysis can normalize $H=I$ using the transformation $\tilde u=H^{1/2}u$.
We retain $H$ in the statements and interpretation of the results, but normalize the loss when convenient in the proofs.

To describe behavior at and near the boundary, consider a sequence of true parameters $\theta_{0,n}\in\Theta_R$ satisfying
\[
\sqrt n(b-A\theta_{0,n})\rightarrow c,
\qquad
c_j\in(-\infty,0]\cup\{-\infty\}.
\]
We assume that the centered Gaussian limit above and convergence of $\hat\Omega$ hold along the local sequences considered. We suppress the subscript $n$ when there is no ambiguity. A coordinate $c_j=0$ corresponds to an exactly binding inequality, a finite $c_j<0$ to an inequality that is locally inside the feasible set at the $n^{-1/2}$ scale, and $c_j=-\infty$ to an asymptotically irrelevant inequality. Coordinates with $c_j=-\infty$ are omitted from the local inequality system. Define
\[
U(c)=\{u\in\mathbb{R}^k:Au\geq c\}.
\]
Under these assumptions, the local feasible sets converge and standard constrained argmin arguments \citep{Geyer:94,Andrews:99} imply that
\[
\sqrt n(\hat\theta_r(\hat\Omega)-\theta_{0,n})
\xrightarrow{d}
\mathcal V_r(\Omega,c),
\]
where
\[
\mathcal V_r(\Omega,c)
=
\argmin_{u\in U(c)}\|u-\mathcal V_{ur}\|_\Omega^2.
\]
We define local asymptotic risk as risk in the limiting Gaussian experiment,
\[
R_H(\Omega,c)
=
E[\mathcal V_r(\Omega,c)'H\mathcal V_r(\Omega,c)].
\]
Interpreting this quantity as the limit of expected scaled finite-sample loss additionally requires uniform integrability. Unrestricted Gaussian risk is
\[
R_{ur}
=E[\mathcal V_{ur}'H\mathcal V_{ur}]
=\operatorname{tr}(H\Sigma).
\]

We use two types of local parameter sets. The enlarged local-slack set is
\[
C=[-\infty,0]^d.
\]
This set is deliberately broad: for a general system of inequalities, not every vector in $C$ has to arise from a feasible sequence of parameters in the original parameter space. For the criteria used below, this broad definition matters for the local no-harm requirement because it makes the requirement conservative. The exact-boundary criterion instead uses only boundary faces that are actually attainable. To define those faces, let $\mathcal J_{\mathrm{att}}$ be the collection of nonempty subsets $J\subseteq\{1,\ldots,d\}$ for which there exists $\theta\in\Theta_R$ satisfying
\[
A_J\theta=b_J,
\qquad
A_{J^c}\theta>b_{J^c}.
\]
For $J\in\mathcal J_{\mathrm{att}}$, define $c(J)\in\{0,-\infty\}^d$ by $c_j(J)=0$ for $j\in J$ and $c_j(J)=-\infty$ otherwise. For $t=1,\ldots,d$, let
\[
\bar C_t
=
\left\{
 c(J):J\in\mathcal J_{\mathrm{att}},\ |J|\leq t
\right\}.
\]
Thus $\bar C_t$ contains the attainable exact-boundary configurations with at least one and at most $t$ binding inequalities. In particular, $\bar C_1$ contains the attainable single-binding configurations and $\bar C_d$ is the complete attainable exact boundary. Positive row rescaling of an inequality, together with the same rescaling of its bound, leaves $\Theta_R$, attainability, and the projection estimator unchanged.

Projection estimators arise naturally as local approximations to constrained extremum estimators. Let $Q_n:\mathbb{R}^k\rightarrow\mathbb{R}$ be twice continuously differentiable and define
\[
\hat\theta_{ur}=\argmin_{\theta\in\mathbb{R}^k}Q_n(\theta).
\]
A second-order expansion around $\hat\theta_{ur}$ yields, uniformly over $\|\theta-\theta_0\|=O(n^{-1/2})$ under standard regularity conditions,
\[
Q_n(\theta)
=
Q_n(\hat\theta_{ur})
+
\frac12(\theta-\hat\theta_{ur})'
\hat\Omega
(\theta-\hat\theta_{ur})
+o_p(n^{-1}),
\]
where $\hat\Omega$ is the local curvature of the criterion. Hence the constrained extremum estimator is locally equivalent to projection using this geometry. Whenever the researcher has flexibility in specifying the criterion, for example through a GMM weighting matrix, the choice can affect the constrained estimator. This is the case even when it has no effect on the unrestricted estimator, as in the just-identified GMM setting.

\subsection{Motivating examples}
\label{subsec:examples}

We next illustrate several settings in which the projection geometry either arises from a familiar estimation criterion or can be chosen directly after estimating an economic object.

\begin{example}[Maximum likelihood estimation]
\label{ex:ml}
Let $\{Y_i\}_{i=1}^n$ be a random sample with density $f_Y(y,\theta_0)$ and
\[
Q_n(\theta)=-\frac1n\sum_{i=1}^n\log f_Y(Y_i,\theta).
\]
The local projection metric is the observed information,
\[
\hat\Omega
=-\frac1n\sum_{i=1}^n
\frac{\partial^2}{\partial\theta\partial\theta'}
\log f_Y(Y_i,\hat\theta_{ur}).
\]
Under standard information-equality conditions, $\hat\Omega$ converges to the inverse of the asymptotic covariance matrix of the unrestricted MLE. Thus constrained maximum likelihood naturally generates inverse-covariance projection in regular models.
\end{example}

\begin{example}[Flexible parametric instrumental variable estimation]
\label{ex:fnc_iv}
Motivated by the setup of \citet{chetverikov2017nonparametric}, suppose
\[
Y=g(X)+\varepsilon,
\qquad
E[\varepsilon\mid Z]=0.
\]
For a vector of basis functions, $b^m(x) = (b_1(x), \ldots, b_m(x))'$, suppose we can write $g(x)=b^m(x)'\theta_0$. Let
\[
B_X=(b^m(X_1),\ldots,b^m(X_n))',
\qquad
B_Z=(b^p(Z_1),\ldots,b^p(Z_n))',
\]
with $p\geq m$. For a positive-definite $p\times p$ matrix $W$, a quadratic IV/GMM criterion takes the form
\[
Q_n(\theta)
=
(Y-B_X\theta)'B_ZWB_Z'(Y-B_X\theta),
\]
which generates the projection metric
\[
\hat\Omega=B_X'B_ZWB_Z'B_X
\]
up to an irrelevant scalar normalization. Thus the GMM weighting matrix determines the geometry used to enforce restrictions on $g$. When the model is just identified, alternative positive-definite choices of $W$ generate the same unrestricted estimator but generally different constrained estimators. With the efficient GMM weight, the local curvature is the inverse asymptotic covariance geometry under the usual regularity conditions.

If performance is evaluated by integrated squared error with nonnegative weight $w(x)$,
\[
\int_0^1(\hat g(x)-g(x))^2w(x)\,dx
=(\hat\theta-\theta_0)'H(\hat\theta-\theta_0),
\]
where
\[
H_{j\ell}=\int_0^1b_j(x)b_\ell(x)w(x)\,dx.
\]
This example makes the three geometries particularly transparent: the IV criterion determines $\Omega$, sampling uncertainty determines $\Sigma$, and integrated squared error determines $H$. We return to this setting in the flexible-IV simulations and the gasoline-demand application.
\end{example}

\begin{example}[Monotone treatment effects]
\label{ex:treatment}
Suppose individuals are assigned to a control group $D=0$ or one of $k$ ordered treatment intensities $D=1,\ldots,k$. Let
\[
\theta_{0,d}=E[Y(d)]-E[Y(0)],
\qquad d=1,\ldots,k,
\]
and impose the ordered-average-treatment-effect restrictions
\[
0\leq\theta_{0,1}\leq\theta_{0,2}\leq\cdots\leq\theta_{0,k}.
\]
These restrictions are linear in $\theta_{0}=(\theta_{0,1},\ldots,\theta_{0,k})'$.

The unrestricted estimator can be written as an exactly identified GMM estimator. With known assignment probabilities $\pi_d=P(D=d)$, define the $k$-vector $\psi_i$ by
\[
\psi_{i,d}
=
\frac{\1\{D_i=d\}Y_i}{\pi_d}
-
\frac{\1\{D_i=0\}Y_i}{\pi_0},
\qquad d=1,\ldots,k.
\]
Under random assignment, $E[\psi_i-\theta_0]=0$. Hence,
\[
\hat\theta_{ur}=\frac1n\sum_{i=1}^n\psi_i
\]
solves the unrestricted sample moments, and for any positive-definite GMM weight $W$ the constrained criterion is exactly
\[
(\theta-\hat\theta_{ur})'W(\theta-\hat\theta_{ur})
\]
up to a constant. The constrained GMM estimator is therefore the $W$-projection of the unrestricted treatment-effect estimator onto the monotonicity cone. Replacing assignment probabilities by sample shares gives the usual difference-in-means estimator, which fits the same projection framework with its corresponding covariance matrix. Under random group counts, its first-order expansion generally differs from that of the estimator using known assignment probabilities.

This design also generates nontrivial covariance geometry for a simple economic reason. Every treatment effect uses the same estimated control mean, so the coordinates of $\hat\theta_{ur}$ share a common sampling component. Efficient GMM therefore uses $W=\Sigma^{-1}$, while other GMM weights generate alternative projection geometries without changing the exactly identified unrestricted estimator. Section~\ref{subsec:sim_treatment} studies this example in detail.
\end{example}

\begin{example}[Shape-restricted impulse responses]
Let
\[
\theta_0=(\beta_0,\beta_1,\ldots,\beta_{k-1})'
\]
collect an impulse-response function estimated by local projections, a VAR, or another
time-series procedure, and suppose
\[
\sqrt{n}\bigl(\hat\theta_{ur}-\theta_0\bigr)
\overset{d}{\rightarrow}N(0,\Sigma).
\]
Economic considerations can imply a variety of linear shape restrictions across horizons.
For example, sign restrictions take the form
\[
\beta_h\geq 0,
\qquad h=0,\ldots,k-1,
\]
while monotone decay can be imposed through
\[
\beta_h-\beta_{h+1}\geq 0,
\qquad h=0,\ldots,k-2.
\]
A hump-shaped response with a prespecified peak horizon $h^*$ can be imposed by requiring
the response to increase up to $h^*$ and decrease afterward. Bounds on cumulative responses,
such as
\[
\underline b_h
\leq
\sum_{j=0}^{h}\beta_j
\leq
\overline b_h,
\]
are linear restrictions as well. Each of these restrictions, and combinations of them, can
therefore be written in the form $A\theta\geq b$.

A restricted impulse-response estimate can then be defined directly by
\[
\hat\theta_r(\Omega)
=
\arg\min_{\theta:A\theta\geq b}
\|\theta-\hat\theta_{ur}\|_\Omega^2.
\]
Here the projection need not be inherited from the criterion that produced
$\hat\theta_{ur}$, but it can be a separate post-estimation choice. Because impulse-response
estimates are typically correlated across horizons and the loss may assign different weights
to different horizons, the sampling covariance $\Sigma$, the loss matrix $H$, and the
projection geometry $\Omega$ can all be distinct. This example illustrates that the choice
of projection geometry also arises naturally outside constrained extremum estimation.
\end{example}

The examples emphasize two complementary sources of projection geometry. In likelihood and GMM problems, the geometry can be inherited from the estimation criterion itself. In other applications, such as shape-restricted impulse responses, projection is a separate post-estimation step. In both cases, the same feasible set can be imposed using different projection geometries, and the resulting risk can differ substantially when projection is active.

To illustrate the effects of the projection geometry, Figure~\ref{fig:linear_iv_projection_geometries} provides a simple two-dimensional example. The data are generated from a just-identified linear IV model with two endogenous regressors and two excluded instruments. The true value is $\theta_0=(0,0)'$, and the maintained restriction is $\theta_2\geq0$.

\begin{figure}[t!]
	\centering
	\caption{Projection geometry in a two-dimensional just-identified IV design}
	\label{fig:linear_iv_projection_geometries}
	\includegraphics[width=0.84\textwidth]{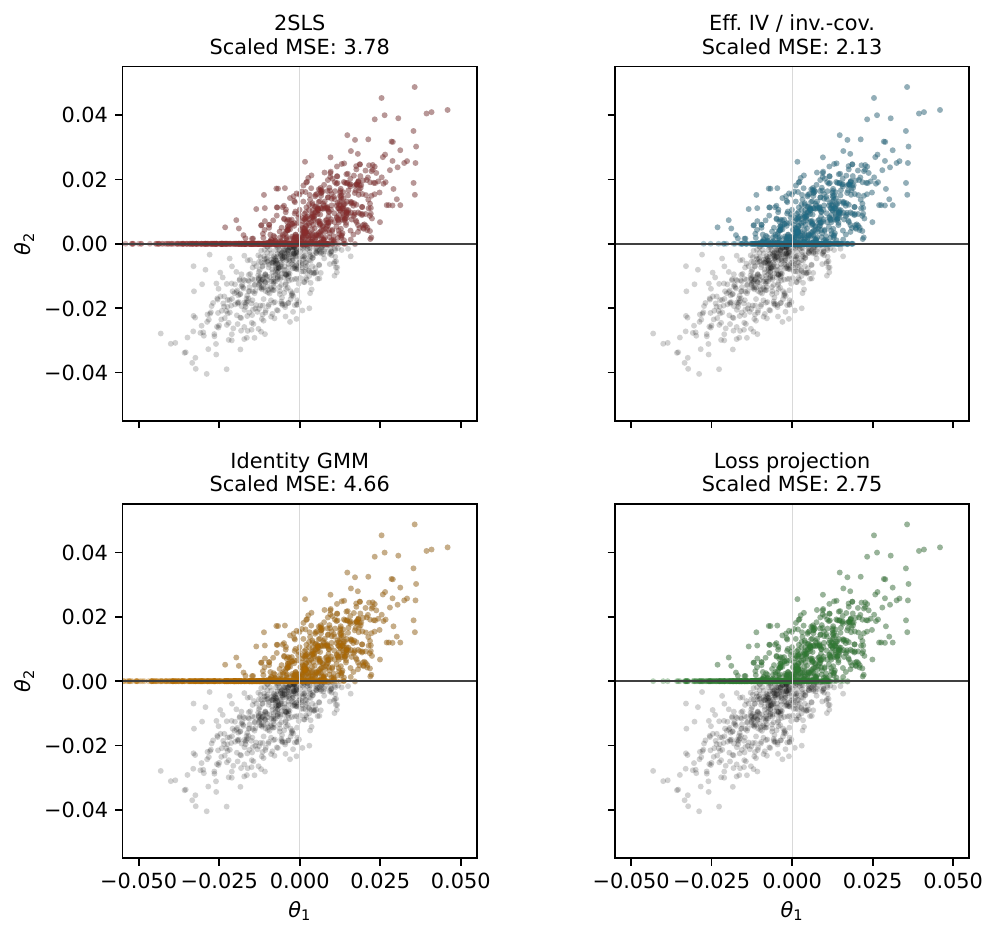}
	\begin{minipage}{0.90\linewidth}
		\vspace{1mm}
		{\footnotesize\begin{singlespace}
				Notes: Black points are unrestricted IV estimates from the same simulated samples; colored points are their projections onto $\{\theta:\theta_2\geq0\}$. The panels differ only in the projection geometry named in the panel title. Scaled MSE is MSE multiplied by $n$.
		\end{singlespace}}
	\end{minipage}
\end{figure}

Since the model is just identified, changing the GMM weight matrix leaves the unrestricted IV estimator unchanged but changes the geometry used to project infeasible estimates back to the feasible set. The figure projects the same unrestricted estimates onto this half-space using the geometry implied by 2SLS, the inverse-covariance geometry implied by efficient GMM, and the geometry implied by the identity GMM weight matrix. The figure also reports Euclidean projection, $\Omega=I$, which coincides with the loss geometry in this example. The plots show draws of the unrestricted estimator, the different restricted estimators, and the MSE for each geometry, scaled by the sample size. This example illustrates that different projection geometries can yield very different restricted estimates and MSEs.

\section{Boundary risk and analytical results} \label{sec:results}
This section develops the criterion used to compare projection geometries and presents the main analytical results. We first explain why minimaxity over the enlarged local-slack set is not informative for this purpose and why a finite neighborhood of the boundary introduces an additional tuning choice. We then study exact-boundary risk, starting with one or two binding inequalities before moving to the general case.

\subsection{Boundary minimaxity}
\label{subsec:boundary_criterion}

Recall from Section~\ref{subsec:setup} that $C=[-\infty,0]^d$ denotes the enlarged local-slack set and that $\bar C_t$ contains the attainable exact-boundary configurations with at most $t$ binding inequalities. In particular, $\bar C_d$ is the complete attainable exact boundary.

To develop a minimax criterion, a first possibility would be to minimize worst-case risk over the enlarged local-slack set. This criterion is well defined but does not identify projection geometry in an informative way. When all constraints are asymptotically irrelevant and $
c=(-\infty,\ldots,-\infty)$, we have $\mathcal V_r(\Omega,c)=\mathcal V_{ur}$ for every $\Omega$, so the worst-case risk of every projection geometry is at least $R_{ur}$. At the same time, projection in the loss metric $H$ satisfies the pathwise contraction property
\[
\|\mathcal V_r(H,c)\|_H^2\leq\|\mathcal V_{ur}\|_H^2
\]
for every realization. Hence the minimax value over $C$ is exactly $R_{ur}$.

\begin{lemma}
	\label{lemma:minimax_entire}
Under the setup in Section~\ref{subsec:setup},
\[
\inf_{\Omega\succ0}\sup_{c\in C}R_H(\Omega,c)=R_{ur},
\]
and $\Omega=H$ is minimax. Any other projection geometry whose risk never exceeds $R_{ur}$ over $C$ is minimax as well.
\end{lemma}

The lemma explains why minimaxity over the enlarged local-slack set is not our primary criterion. The minimax value is pinned down by states in which the constraints have no first-order effect. A criterion designed to choose how violations should be corrected should instead give weight to configurations in which projection matters.

A second possibility is to restrict attention to a one-sided finite band inside the feasible side of the boundary. Even in the single-halfspace problem, however, the resulting minimax problem depends on the width of that band, as illustrated in the following lemma. 

\begin{lemma}
\label{lemma:minimax_local}
Consider one inequality $a'u\geq c$ under quadratic loss $u'Hu$, where $H\succ0$, and normalize $a'\Sigma a=1$. Write $c=-\delta$, with $\delta\geq0$, and for $\bar\delta>0$ consider
\[
\inf_{\Omega\succ0}\max_{\delta\in[0,\bar\delta]}R_H(\Omega,-\delta).
\]
\begin{enumerate}
\item For every projection geometry $\Omega$, the worst-case risk over $\delta\in[0,\bar\delta]$ is attained at one of the endpoints of the interval, that is, at $\delta=0$ or $\delta=\bar\delta$.
\item There exists a minimax projection geometry for which $\delta=\bar\delta$ is a worst-case point.
\item Unless $\Sigma a$ is proportional to $H^{-1}a$, inverse-covariance projection is not minimax for any $\bar\delta>0$.
\item For all $\delta$ in a sufficiently small right-neighborhood of zero, a pointwise optimal inverse projection metric can be chosen such that
\[
\Omega^{-1}(\delta)
=
\frac{w_1(\delta)}{w_1(\delta)+w_2(\delta)}\Sigma
+
\frac{w_2(\delta)}{w_1(\delta)+w_2(\delta)}
\frac{H^{-1}}{a'H^{-1}a},
\]
where $w_1(\cdot)$ and $w_2(\cdot)$ are continuous functions satisfying $w_1(0)=1$ and $w_2(0)=0$.
\end{enumerate}
\end{lemma}

The lemma explains why minimaxity over a finite neighborhood does not provide a tuning-free criterion. Part 1 reduces the worst-case comparison to the two endpoints, while part 2 shows that the researcher-chosen outer edge $\bar\delta$ can be directly relevant for the minimax solution. Part 3 sharpens this point: except when the covariance and loss projection directions coincide, inverse-covariance projection ceases to be minimax as soon as the neighborhood has positive width.

Part 4 gives additional structure to the pointwise problem. A pointwise optimal inverse projection metric can be chosen as an affine combination of the covariance geometry $\Sigma$ and the normalized inverse loss geometry $H^{-1}/(a'H^{-1}a)$. At the exact boundary the second weight vanishes, so this choice becomes inverse-covariance projection, $\Omega^{-1}=\Sigma$.

The coefficients in part 4 sum to one, but they need not both be positive. The result nevertheless shows that the two geometries $\Sigma$ and $H^{-1}$ govern the single-halfspace pointwise optimum. Section~\ref{subsec:path} uses these two theoretically distinguished objects to construct a simple positive-definite path for the multi-constraint selection problem.

Although Lemma~\ref{lemma:minimax_local} is stated in the local Gaussian experiment used throughout the paper, it also has an exact finite-sample interpretation under normality. If $\hat\theta_{ur}\sim N(\theta_0,\Sigma/n)$ and the restricted estimator is its exact quadratic projection, the same calculation applies after expressing the width of a one-sided finite parameter strip in root-$n$, standard-error units. Thus the dependence of the optimal projection geometry on the chosen neighborhood is not an artifact of the asymptotic notation.

The lemma shows that a finite-band criterion requires a choice of distance to the boundary that affects the projection geometry itself. With several inequalities, one must additionally choose the shape of the neighborhood and potentially relative scales across constraints. Exact binding instead identifies a face of the feasible set without introducing a distance-to-boundary tuning parameter. We therefore define boundary risk by
\[
B(\Omega)=\max_{c\in\bar C_d}R_H(\Omega,c)
\]
and call a projection geometry boundary minimax if it attains $\inf_{\Omega\succ0}B(\Omega)$ over the class of positive-definite metrics.
We use an infimum because positive scalar multiples of $\Omega$ induce the same projection and the positive-definite cone is not compact. We therefore call a geometry boundary minimax only when it attains this infimum.

This criterion should not be interpreted as assuming that the researcher knows the true parameter lies on a particular boundary face. The estimator remains the inequality projection defined over the full feasible parameter space. Boundary configurations are used only to rank projection geometries. This distinction is important: if a binding equality itself were known, an unrestricted decision rule could exploit that equality directly and need not belong to the projection class considered here.

Appendix~\ref{appendix:geometry_illustrations} gives a two-dimensional visualization of this issue: a geometry that is attractive when one face is active can be unattractive at a corner, and conversely.

\subsection{One locally relevant inequality}
\label{subsec:one_const}

We first consider local configurations in which at most one inequality is relevant. Part~1 concerns attainable exact-boundary configurations with one binding inequality, while Part~2 allows the single relevant inequality to have arbitrary finite local slack.

\begin{theorem}
\label{thm:one_const}
For any maintained system of linear inequalities satisfying the setup assumptions,
\begin{enumerate}
\item For every $c\in\bar C_1$,
\[
\Sigma^{-1}\in\argmin_{\Omega\succ0}R_H(\Omega,c).
\]
\item For every $c\in C$ with at most one finite coordinate,
\[
R_H(\Sigma^{-1},c)\leq R_{ur}.
\]
\end{enumerate}
\end{theorem}

Part~1 states that inverse-covariance projection is pointwise optimal at every attainable one-binding exact-boundary configuration. Because the statement holds pointwise over $\bar C_1$, it also implies the boundary-minimax result
\[
\Sigma^{-1}\in\argmin_{\Omega\succ0}\max_{c\in\bar C_1}R_H(\Omega,c).
\]
Part~2 provides a robustness statement that inverse-covariance projection has no larger $H$-risk than unrestricted estimation whenever at most one local slack coordinate is finite.

The exact penalty from using a different projection direction is also face-specific. Fix $c\in\bar C_1$, let $j$ denote its unique binding coordinate, and let $a_j$ be the corresponding constraint normal. Define
\[
d_{\Omega,j}=\frac{\Omega^{-1}a_j}{a_j'\Omega^{-1}a_j},
\qquad
d_j^*=\frac{\Sigma a_j}{a_j'\Sigma a_j}.
\]
Then
\[
R_H(\Omega,c)-R_H(\Sigma^{-1},c)
=
\frac{a_j'\Sigma a_j}{2}(d_{\Omega,j}-d_j^*)'H(d_{\Omega,j}-d_j^*).
\]
Thus the relevant object on each one-binding face is the projection direction $d_j^*$, rather than the full metric. Any positive-definite projection metric that generates this direction has the same risk at that configuration, and $\Sigma^{-1}$ is the canonical metric that does so simultaneously for every $c\in\bar C_1$.

To see the mechanism, fix a local configuration with one finite coordinate $j$, suppress the asymptotically irrelevant inequalities, write $a=a_j$, and denote the scalar finite slack by $\gamma\leq0$. The reduced local restriction is $a'u\geq\gamma$. Conditional normality gives
\[
E[\mathcal V_{ur}\mid a'\mathcal V_{ur}=x]
=
\frac{\Sigma a}{a'\Sigma a}x,
\]
while the conditional covariance does not depend on $x$. Thus, conditional on the scalar index $a'\mathcal V_{ur}$, the mean of the unrestricted estimator moves only along the projection direction generated by inverse covariance. At exact binding, $\gamma=0$, projecting violations in this direction centers the restricted conditional mean at the boundary value. The proof in Appendix~\ref{appendix:proofs} formalizes the resulting risk comparison and the local no-harm result.

There is also a classical route to Part~1. For $c\in\bar C_1$ with unique binding coordinate $j$, Gaussian symmetry gives
\[
R_H(\Omega,c)=\frac12R_{ur}+\frac12R_{H,j}^{=}(\Omega),
\]
where $R_{H,j}^{=}(\Omega)$ is the risk from projecting onto the corresponding equality restriction $a_j'u=0$. Hence Part~1 can alternatively be obtained face by face from the classical optimality of inverse-covariance projection under linear equality restrictions \citep{Theil1971}. We retain the direct conditional-normal argument because it yields the exact penalty above and also proves the finite-slack result in Part~2.

In the canonical Gaussian sequence model with Euclidean loss, $\Sigma\propto I$ and $H\propto I$, so the inverse-covariance, loss, and Euclidean projection benchmarks coincide up to scale. Other projection geometries can still differ. In particular, identity covariance alone does not equate these benchmarks when the loss matrix is non-Euclidean.

\subsection{Up to two locally relevant inequalities}
\label{subsec:two_const}

We next consider local configurations with at most two relevant inequalities. Part~1 concerns the exact-boundary set $\bar C_2$, while Part~2 allows arbitrary finite local slacks in up to two coordinates. At a two-binding corner, inverse-covariance projection need not be pointwise optimal, but it remains minimax over the exact-boundary configurations in $\bar C_2$.

\begin{theorem}
\label{thm:two_const}
For any maintained system of linear inequalities satisfying the setup assumptions,
\begin{enumerate}
\item
\[
\Sigma^{-1}\in\argmin_{\Omega\succ0}\max_{c\in\bar C_2}R_H(\Omega,c).
\]
\item For every $c\in C$ with at most two finite coordinates,
\[
R_H(\Sigma^{-1},c)\leq R_{ur}.
\]
\end{enumerate}
\end{theorem}

Part 1 is our geometry-choice result. For linearly independent relevant normals, Part 2 follows from the classical restricted-normal domination result of \citet{RuedaSalvador1995}. If the relevant normals are linearly dependent, the local feasible set reduces to a halfspace or an interval, as shown in Appendix~\ref{appendix:proofs}.

The intuition for Part 1 is that a two-binding corner does not create a new worst-case configuration under inverse-covariance projection. The one-binding cases already give a lower bound on the minimax risk, because inverse covariance is pointwise optimal for each single binding inequality. The proof shows that when both inequalities bind, the risk under inverse covariance is no larger than the larger of the two corresponding one-binding risks.

Thus the two-binding face can be controlled by the one-binding benchmarks. Because each one-binding benchmark is already minimized by inverse covariance, any competing geometry must weakly increase at least one of these lower-bound risks. Inverse covariance is therefore minimax over exact-boundary configurations with at most two binding inequalities. The full proof is given in Appendix~\ref{appendix:proofs}.

\subsection{A sufficient condition for boundary minimaxity}
\label{subsec:suff_general}

The two-binding argument suggests a more general sufficient condition. The attainable one-binding configurations continue to give a lower bound on the boundary-minimax value because inverse covariance is pointwise optimal for each of them. It is therefore sufficient that the worst attainable exact-boundary risk under inverse covariance is already attained in a one-binding configuration.

\begin{theorem}
\label{thm:suff_cond}
Under the setup assumptions, if
\[
\max_{c\in\bar C_d}R_H(\Sigma^{-1},c)
=
\max_{c\in\bar C_1}R_H(\Sigma^{-1},c),
\]
then
\[
\Sigma^{-1}\in\argmin_{\Omega\succ0}\max_{c\in\bar C_d}R_H(\Omega,c).
\]
\end{theorem}

The condition can be checked for a given constraint system, covariance matrix, and loss matrix by evaluating the finitely many attainable exact-boundary configurations. It is automatically satisfied when no attainable exact-boundary configuration has more than two binding inequalities. In that case, $\bar C_d=\bar C_2$, and the proof of Theorem~\ref{thm:two_const} shows that every two-binding risk under inverse covariance is bounded by a corresponding one-binding risk. The next example shows that the condition can fail once three inequalities can bind.

\subsection{A three-constraint numerical counterexample}
\label{subsec:three_counterexample}

We next give a numerical counterexample showing that the sufficient condition in Theorem~\ref{thm:suff_cond} can fail when three inequalities bind. Let $\mathcal V_{ur}\sim N(0,I_3)$, so inverse covariance corresponds to the metric $I_3$. Consider the three inequality normals
\[
a_1=(\tau,0,\rho)',\qquad
a_2=\left(-\frac{\tau}{2},\frac{\sqrt3\tau}{2},\rho\right)',\qquad
a_3=\left(-\frac{\tau}{2},-\frac{\sqrt3\tau}{2},\rho\right)',
\]
where $\rho=0.9$ and $\tau=\sqrt{1-\rho^2}=\sqrt{0.19}$. Every nonempty subset of these constraints is attainable as an exact active set. We evaluate risk under
\[
H_\varepsilon=\operatorname{diag}(10^{-4},10^{-4},1),
\]
which places almost all weight on the third, vertical coordinate. We compare inverse covariance with the alternative metric $\operatorname{diag}(1,1,1.35)$. For each nonempty binding set $S$, define the exact-boundary risk after projecting $\mathcal V_{ur}$ onto the cone defined by the binding inequalities as
\[
R_S(\Omega)
=
E\left[
\Pi_{\{u:a_j'u\geq0,\ j\in S\}}^{\Omega}(\mathcal V_{ur})'
H_\varepsilon
\Pi_{\{u:a_j'u\geq0,\ j\in S\}}^{\Omega}(\mathcal V_{ur})
\right].
\]
Table~\ref{tab:three_constraint_counterexample} reports Monte Carlo estimates of these risks, with simulation standard errors for the fixed-size face comparisons. Those differences are much larger than their Monte Carlo standard errors.

\begin{table}[!htbp]
\centering
\caption{Three-constraint counterexample}
\label{tab:three_constraint_counterexample}
\footnotesize
\begin{tabular}{lccc}
\toprule
Risk summary & \makecell{Inverse covariance\\$I_3$} & \makecell{Alternative metric\\$\operatorname{diag}(1,1,1.35)$} & Difference\\
\midrule
One-binding risk & \makecell{0.592699\\(0.001975)} & \makecell{0.594320\\(0.001974)} & \makecell{-0.001621\\(0.000044)}\\
Two-binding risk & \makecell{0.581214\\(0.001992)} & \makecell{0.569673\\(0.001997)} & \makecell{0.011541\\(0.000039)}\\
Three-binding risk & \makecell{0.612548\\(0.001987)} & \makecell{0.593276\\(0.001994)} & \makecell{0.019272\\(0.000059)}\\
\midrule
Worst boundary risk & 0.612548 & 0.594320 & 0.018228\\
\bottomrule
\end{tabular}
\begin{minipage}{0.92\linewidth}
\vspace{2mm}
{\footnotesize\begin{singlespace}
Notes: The difference column reports the first risk minus the second. By symmetry, all one-binding configurations share a population risk, as do all two-binding configurations. We estimate each common risk by averaging draw-level losses over the corresponding faces. Parentheses report ordinary Monte Carlo standard errors from 300,000 iid Gaussian draws, paired across metrics. The last row takes the largest of the three size-specific estimates for each metric. We do not attach a selected mean's standard error to these maxima. Appendix~\ref{appendix:counterexample_details} reports all seven individual face estimates.
\end{singlespace}}
\end{minipage}
\end{table}

The table shows that inverse covariance is best on the one-binding faces, as required by Theorem~\ref{thm:one_const}, but its worst boundary risk occurs when all three inequalities bind. Hence
\[
\max_{c\in\bar C_3}R_{H_\varepsilon}(I_3,c)
>
\max_{c\in\bar C_1}R_{H_\varepsilon}(I_3,c),
\]
so the sufficient condition in Theorem~\ref{thm:suff_cond} fails. The alternative metric accepts a small increase in one-binding risk but reduces the three-binding risk enough to lower the worst boundary risk. Thus inverse-covariance projection is not boundary minimax in this example.

The geometry explains the failure. The three constraint planes are tilted toward the vertical direction that receives almost all weight under $H_\varepsilon$. With all three inequalities imposed, increasing the vertical projection weight changes the active-face projection directions and reduces multi-binding risk enough to offset the small one-binding loss. The minimax geometry can therefore depend on the interaction among the covariance, loss, and constraint geometries once three or more inequalities bind.

This example shows that Part~1 of Theorem~\ref{thm:two_const} is sharp with respect to the number of binding inequalities on an attainable exact-boundary face: without additional conditions, inverse-covariance boundary minimaxity does not extend from $\bar C_2$ to $\bar C_3$. At the same time, Theorem~\ref{thm:suff_cond} remains useful because it identifies higher-dimensional problems in which the one- and two-binding logic continues to apply.

\subsection{Inverse-covariance projection can increase risk}
\label{subsec:no_harm_example}

The local no-harm guarantee in Part~2 of Theorem~\ref{thm:two_const} does not extend automatically once several inequalities interact. This point is related to \citet{ShinozakiChang1999}, who show that the restricted Gaussian MLE under positive-orthant restrictions need not dominate the unrestricted estimator for every linear functional once the number of means reaches five. The following exact calculation gives the corresponding failure in our quadratic-risk framework.

Let $k=d=5$, $\Sigma=I_5$, impose the coordinatewise restrictions $\theta_j\geq0$, and evaluate risk using
\[
H=0.05I_5+\mathbf 1_5\mathbf 1_5'.
\]
At the boundary configuration $c=\mathbf 0_5$, inverse-covariance projection is Euclidean projection onto the positive orthant. Thus, for $\mathcal V_{ur}\sim N(0,I_5)$,
\[
\mathcal V_r(I_5,\mathbf 0_5)=(\mathcal V_{ur,1,+},\ldots,\mathcal V_{ur,5,+})',
\qquad \mathcal V_{ur,j,+}=\max\{\mathcal V_{ur,j},0\}.
\]
The unrestricted risk is
\[
R_{ur}=\operatorname{tr}(H)=5.25.
\]
Using $E[\mathcal V_{ur,j,+}]=1/\sqrt{2\pi}$ and $E[\mathcal V_{ur,j,+}^2]=1/2$,
\begin{align*}
R_H(I_5,\mathbf 0_5)
&=0.05E\!\left[\sum_{j=1}^5\mathcal V_{ur,j,+}^2\right]
+E\!\left[\left(\sum_{j=1}^5\mathcal V_{ur,j,+}\right)^2\right]\\
&=0.05\frac52+\frac52+\frac{10}{\pi}
\approx5.808.
\end{align*}
Thus
\[
\frac{R_H(\Sigma^{-1},\mathbf 0_5)}{R_{ur}}\approx1.106.
\]
Inverse-covariance projection therefore increases asymptotic risk by about $10.6\%$ at this boundary configuration. The example is useful because the sampling covariance is proportional to the identity, so the risk increase comes from the interaction of several constraints with a non-Euclidean loss. It motivates the separate no-harm requirement used in Section~\ref{sec:selection}.

\section{Selecting projection geometry in practice} \label{sec:selection}
Under our boundary-risk and no-harm criteria, the results in Section~\ref{sec:results} determine the projection geometry when the maintained system contains at most two inequalities. For larger maintained systems, configurations with one or two relevant inequalities are still analytically controlled, but exact-boundary faces with three or more binding inequalities can overturn boundary minimaxity and local configurations with three or more finite slack coordinates can violate the no-harm condition, as Sections~\ref{subsec:three_counterexample} and~\ref{subsec:no_harm_example} show. This section develops a practical criterion for those cases. The main idea is to separate two objectives: efficiency at exact-boundary configurations and protection over the enlarged local-slack set.

\subsection{Boundary risk and the local no-harm condition}
\label{subsec:population_selection}

Recall the local asymptotic risk
\[
R_H(\Omega,c)=E[\mathcal V_r(\Omega,c)'H\mathcal V_r(\Omega,c)]
\]
and unrestricted risk
\[
R_{ur}=\operatorname{tr}(H\Sigma).
\]
Recall from Section~\ref{subsec:setup} that $C$ is the enlarged local-slack set and $\bar C_d$ is the complete attainable exact boundary. We define worst-case exact-boundary risk by
\[
B(\Omega)=\max_{c\in\bar C_d}R_H(\Omega,c)
\]
and the local no-harm ratio by
\[
G(\Omega)=\sup_{c\in C}\frac{R_H(\Omega,c)}{R_{ur}}.
\]
We say that a projection geometry satisfies the \emph{local no-harm condition} if
\begin{equation}
G(\Omega)\leq1.
\tag{NH}\label{eq:no_harm}
\end{equation}
Because the all-interior configuration belongs to $C$ and gives $\mathcal V_r(\Omega,c)=\mathcal V_{ur}$, $G(\Omega)\geq1$ for every projection geometry. Thus condition~\eqref{eq:no_harm} is equivalent, at the population level, to $G(\Omega)=1$.
Our population selection problem is
\begin{equation}
\inf_{\Omega\succ0}B(\Omega)
\qquad\text{subject to}\qquad
G(\Omega)\leq1.
\tag{P}\label{eq:population_choice}
\end{equation}

The objective and constraint use different local comparisons. Exact-boundary configurations provide a tuning-free benchmark for the efficiency of the projection geometry. Finite local slacks enter through the no-harm condition, which rules out a geometry if it performs worse than ignoring the restrictions somewhere in the enlarged local-slack set. Because $C$ is enlarged, this is a conservative no-harm requirement. Thus $B$ ranks projection geometries on the boundary, while $G$ screens out geometries that sacrifice the unrestricted benchmark elsewhere in the local experiment.

When the maintained system has $d\leq2$, Theorem~\ref{thm:two_const} provides both ingredients of Problem~\eqref{eq:population_choice}: inverse covariance minimizes boundary risk and satisfies $G(\Sigma^{-1})=1$. Hence $\Sigma^{-1}$ solves Problem~\eqref{eq:population_choice} without further numerical verification.

Problem~\eqref{eq:population_choice} has a feasible projection geometry. Projection in the loss metric, $\Omega=H$, satisfies
\[
\|\mathcal V_r(H,c)\|_H^2\leq\|\mathcal V_{ur}\|_H^2
\]
draw by draw and therefore satisfies \eqref{eq:no_harm}. Inverse covariance, $\Omega=\Sigma^{-1}$, is the other natural benchmark due to the following result.

\begin{proposition}
\label{prop:population_choice}
Suppose the sufficient condition in Theorem~\ref{thm:suff_cond} holds. If $\Sigma^{-1}$ also satisfies the local no-harm condition, then
\[
\Sigma^{-1}\in\argmin_{\Omega:G(\Omega)\leq1}B(\Omega).
\]
Hence inverse-covariance projection solves Problem~\eqref{eq:population_choice}, not only the restricted path problem introduced below.
\end{proposition}

The proposition follows immediately because Theorem~\ref{thm:suff_cond} establishes that $\Sigma^{-1}$ minimizes $B(\Omega)$ without the no-harm constraint, while the additional assumption makes it feasible. This result is useful for maintained systems with more than two inequalities: if the sufficient condition and local no-harm condition can both be verified, no further search is needed for Problem~\eqref{eq:population_choice}. In the ordered-treatment simulations and gasoline-demand application below, these diagnostics support inverse covariance. At the same time, the deliberately constructed counterexamples in Sections~\ref{subsec:three_counterexample} and~\ref{subsec:no_harm_example} show that the diagnostics are not automatic.

\subsection{A one-dimensional implementation}
\label{subsec:path}

Directly solving Problem~\eqref{eq:population_choice} over the positive-definite cone is generally computationally demanding. Lemma~\ref{lemma:minimax_local} gives a useful heuristic for choosing a low-dimensional search family. In the single-halfspace problem, the inverse of a pointwise optimal projection metric can be chosen in the affine span of $\Sigma$ and a normalized version of $H^{-1}$. Because the coefficients depend on the local slack and constraint and need not be nonnegative, this formula does not define a canonical multi-constraint rule, but it identifies $\Sigma$ and $H^{-1}$ as natural endpoints.

Since positive scalar multiples of an inverse metric generate the same projection geometry, we normalize the two path endpoints before interpolation. Let
\[
\bar K_H=\frac{H^{-1}}{\operatorname{tr}(H^{-1})},
\qquad
\bar K_\Sigma=\frac{\Sigma}{\operatorname{tr}(\Sigma)}.
\]
We use the convex path
\begin{equation}
K_s=\Omega_s^{-1}=(1-s)\bar K_H+s\bar K_\Sigma,
\qquad 0\leq s\leq1.
\label{eq:path}
\end{equation}
Thus the endpoint geometries are the loss projection at $s=0$ and the inverse-covariance projection at $s=1$. The compact interval $[0,1]$ is convenient for numerical search and guarantees that all matrices along the path are positive-definite. The trace normalization removes arbitrary positive scaling of the two endpoint matrices and makes the reported value of $s$ reproducible for the maintained parameterization. Thus $s$ is a coordinate on this normalized path, not an invariant feature of the problem.

Along this path, the population boundary-risk selector is
\begin{equation}
 s_{BR}\in\argmin_{s\in[0,1]} B(\Omega_s)
 \qquad\text{subject to}\qquad
 G(\Omega_s)\leq1.
\label{eq:sbr}
\end{equation}
The path is a tractable bridge between interpretable endpoints, not a search over the full positive-definite cone. 

\subsection{Feasible risk evaluation}
\label{subsec:feasible_selector}

In applications, both the covariance matrix and the set of inequalities that are sufficiently close to binding to matter for metric selection are unknown. We use a plug-in covariance matrix $\hat\Sigma$ and a conservative selection of potentially binding inequalities. For the restriction $a_j'\theta\geq b_j$, define
\[
T_j
=
\frac{\sqrt n(a_j'\hat\theta_{ur}-b_j)}
{\sqrt{a_j'\hat\Sigma a_j}}
\]
and set
\[
\hat J=\{j:T_j\leq\kappa\}.
\]
For asymptotic interpretation, one may use a deterministic sequence $\kappa_n\to\infty$ with $\kappa_n/\sqrt n\to0$. Inequalities with finite local slack then satisfy $T_j=O_p(1)$ and are retained, while inequalities fixed in the interior have $T_j$ diverging at the $\sqrt n$ rate and are not included, both with probability approaching one. The simulations use the conservative choice $\kappa=5$.

In feasible calculations, the covariance endpoint in \eqref{eq:path} is formed by replacing $\Sigma$ with $\hat\Sigma$. Thus the reported path value refers to the same normalized interpolation, with the inverse-covariance endpoint estimated from the data.

The selection is used only to reduce the dimension of the risk calculation used to choose the geometry. It is not a pretest for whether a maintained inequality is imposed: after the geometry is chosen, the final estimator still projects onto the full maintained feasible set.

If no inequality is retained, the metric-selection step is skipped and the unrestricted estimator is reported. Since $\kappa\geq0$, an empty retained set implies that the unrestricted estimate satisfies every maintained inequality, so projection leaves it unchanged. False positives increase computation but do not remove restrictions from the estimator. False negatives are more consequential because they omit a local constraint from metric selection.

Conditional on the selected set $\hat J$, standardize the retained local slacks by writing
\[
c_{\hat J}=-D\mu,
\qquad
D=\operatorname{diag}\left(
\sqrt{a_j'\hat\Sigma a_j}:j\in\hat J
\right),
\qquad \mu\geq0.
\]
Thus $\mu_j$ measures distance from the boundary in standard-deviation units. The no-harm criterion allows $\mu_j\in[0,\infty]$. For computation, finite coordinates are searched up to $M$, while infinite coordinates are handled by omitting the corresponding inequalities as described below. Our baseline uses $M=8$. Sensitivity checks with other values of $M$ can be carried out if the maximizer lies on, or very near, the boundary of the reported finite-box calculation.

The boundary criterion only requires enumerating exact faces. Let $\widehat{\mathcal J}_{\mathrm{att}}$ be the collection of nonempty subsets $L\subseteq\hat J$ for which there exists $\theta\in\Theta_R$ such that $a_j'\theta=b_j$ for $j\in L$ and $a_j'\theta>b_j$ for $j\in\hat J\setminus L$. Thus all maintained inequalities remain imposed, while nonretained inequalities are treated as asymptotically irrelevant in the local risk calculation. The estimated boundary criterion is
\[
\hat B(s)
=
\max_{L\in\widehat{\mathcal J}_{\mathrm{att}}}
\frac{\hat R_H(s,\mathbf 0_L;L)}
{\operatorname{tr}(H\hat\Sigma)}.
\]
Here $\hat R_H(s,\mu_L;L)$ is the plug-in risk under $\Omega_s$ with raw slack $-D_L\mu_L$ for the inequalities in $L$, while the other retained inequalities are asymptotically irrelevant. The normalization reports risk relative to unrestricted risk and does not affect the selected path value.

The estimated finite-box local no-harm diagnostic maximizes over lower-dimensional faces of the retained local slack set,
\[
\hat G_M(s)
=
\max\left\{1,\;
\max_{\emptyset\neq L\subseteq\hat J}
\max_{\mu_L\in[0,M]^{|L|}}
\frac{\hat R_H(s,\mu_L;L)}{\operatorname{tr}(H\hat\Sigma)}
\right\}.
\]
Here the inequalities in $L$ are assigned finite standardized slacks and the other retained inequalities are treated as asymptotically irrelevant. Attainability is imposed for the boundary criterion, whereas the no-harm diagnostic keeps the conservative enlarged-set convention.

For modest $|\hat J|$, both $\hat B(s)$ and $\hat G_M(s)$ can be computed efficiently using the dual representation of the projection problem. Exact-boundary enumeration is exponential in $|\hat J|$, and the finite-slack search for each face has dimension $|L|$. The implementation is therefore intended for settings in which the conservative screen leaves a moderate number of potentially relevant inequalities. The dual variables depend on the draw only through $A_L\mathcal V_{ur}$. Conditional Gaussian moments integrate out the remaining components, leaving Gaussian dimension $\operatorname{rank}(A_L)\leq |L|$, with equality when the rows have full rank. Our implementation combines active-set enumeration for the resulting nonnegative quadratic program with scrambled Sobol integration. The use of common draws is described in Appendix~\ref{appendix:selector_details}.

\subsection{Plug-in covariance and metric estimation}
\label{subsec:cov_estimation}

The preceding criterion is implemented with estimated covariance and projection metrics. We first record a general plug-in stability result. We then specialize it to the inverse-covariance boundary diagnostic from Proposition~\ref{prop:population_choice}. For notational simplicity, we work in normalized coordinates in which $H=I$; the general weighted-loss version follows from the coordinate transformation in Section~\ref{subsec:setup}.

Fix positive-definite matrices $\Sigma$ and $\Omega$, and let $\hat\Sigma$ and $\hat\Omega$ be symmetric estimators. On the event that $\hat\Sigma$ and $\hat\Omega$ are positive definite, define the plug-in Gaussian projection
\[
\hat{\mathcal V}_r(\hat\Sigma,\hat\Omega,c)
=
\argmin_{u\in U(c)}
\|u-\hat\Sigma^{1/2}Z\|_{\hat\Omega}^2,
\qquad Z\sim N(0,I_k),
\]
where $Z$ is independent of $(\hat\Sigma,\hat\Omega)$. Outside this event, set the plug-in projection equal to zero. This convention does not affect the probability limits below.
Let $\mathcal V_r(\Omega,c)$ denote the analogous population projection of $\Sigma^{1/2}Z$ using the metric $\Omega$, and let
\[
R_{\Sigma,\Omega}(c)
=
E[\|\mathcal V_r(\Omega,c)\|^2],
\qquad
\hat R_n(c)
=
E[\|\hat{\mathcal V}_r(\hat\Sigma,\hat\Omega,c)\|^2\mid\hat\Sigma,\hat\Omega].
\]

\begin{theorem}
\label{thm:plugin_boundary_stability}
Work in normalized coordinates in which $H=I$. Suppose $\hat\Sigma\xrightarrow{p}\Sigma$, $\hat\Omega\xrightarrow{p}\Omega$, and
$P(\hat\Sigma\succ0,\hat\Omega\succ0)\to1$, where $\Sigma\succ0$ and $\Omega\succ0$. Then
\[
\sup_{c\in C}
\left|\hat R_n(c)-R_{\Sigma,\Omega}(c)\right|
=o_p(1).
\]
\end{theorem}

\begin{corollary}
\label{cor:plugin_boundary_stability}
In the same normalized coordinates, suppose $\hat\Sigma\xrightarrow{p}\Sigma\succ0$ and $P(\hat\Sigma\succ0)\to1$. Set $\Omega=\Sigma^{-1}$ and $\hat\Omega=\hat\Sigma^{-1}$ on the event that $\hat\Sigma$ is positive definite, using the zero convention above outside this event. Define the population gap for the inverse-covariance boundary diagnostic by
\[
\Delta
=
\max_{c\in\bar C_d}R_{\Sigma,\Sigma^{-1}}(c)
-
\max_{c\in\bar C_1}R_{\Sigma,\Sigma^{-1}}(c)
\geq0
\]
and its plug-in analogue
\[
\hat\Delta_n
=
\max_{c\in\bar C_d}\hat R_n(c)
-
\max_{c\in\bar C_1}\hat R_n(c).
\]
Then
\[
|\hat\Delta_n-\Delta|=o_p(1).
\]
In particular, if the sufficient condition in Theorem~\ref{thm:suff_cond} fails strictly, so that $\Delta>0$, then $\hat\Delta_n\xrightarrow{p}\Delta$ and the strict failure is detected with probability approaching one.
\end{corollary}

The theorem justifies plug-in Gaussian risk calculations for any fixed limiting projection geometry. The corollary applies this result to the inverse-covariance endpoint $s=1$ and its boundary diagnostic. The same argument gives uniform convergence over the compact path in \eqref{eq:path}, because the population and feasible path matrices are uniformly positive definite and converge uniformly in $s$. This uniformity justifies plug-in risk evaluation along the path. Consistency of the constrained selector itself would also require stability of the retained constraints and estimated feasible sets, together with separation of the population minimizer. Feasible-set stability is nontrivial because population feasibility requires $G=1$, with no strict $G<1$ margin. The same conclusions for risk evaluation apply to normalized ratios because $\operatorname{tr}(\hat\Sigma)\to\operatorname{tr}(\Sigma)$ in normalized coordinates and $\operatorname{tr}(H\hat\Sigma)\to\operatorname{tr}(H\Sigma)$ after transforming back.

The gap $\Delta$ is zero exactly when the sufficient condition in Theorem~\ref{thm:suff_cond} holds, in which case $\hat\Delta_n \xrightarrow{p} 0$. If $\Delta=0$, finite-sample calculations should therefore report the numerical plug-in gap rather than treat equality as a formal decision. Passing the boundary diagnostic is not by itself the final recommendation. We use it together with the local no-harm calculation: if inverse covariance is boundary minimax and also satisfies local no-harm, Proposition~\ref{prop:population_choice} establishes that it solves the population selection problem. Otherwise, we search along the path in \eqref{eq:path}. The simulations in Section~\ref{sec:simulations} illustrate selected-set screening, covariance estimation, finite-$M$ approximation, and the resulting feasible selector.

\subsection{Inverse-covariance projection as a default}
The selection criterion is useful when the researcher wants to adapt the projection geometry to the particular combination of sampling uncertainty, loss, and constraints. If one instead wants to use a prespecified geometry without running the diagnostic search, inverse covariance is the natural default.

In the exact Gaussian experiment $\hat\theta_{ur}\mid\theta\sim N(\theta,\Sigma/n)$, projection using $\Sigma^{-1}$ coincides with constrained maximum likelihood. Under a prior that is flat on $\Theta_R$ and zero outside it, the same projection is also the posterior mode. More generally, regular constrained maximum likelihood is locally equivalent to projecting the unrestricted MLE using the inverse asymptotic covariance matrix, and efficient GMM leads to the same inverse-covariance geometry in the corresponding local approximation.

These likelihood and extremum-estimation arguments are reinforced by the risk results in Section~\ref{sec:results}. Inverse covariance is pointwise optimal at every attainable one-binding exact-boundary configuration, is boundary minimax over $\bar C_2$, and is locally no-harm whenever at most two inequalities have finite local slack. Hence, when the maintained system has at most two inequalities, it solves Problem~\eqref{eq:population_choice}. With more inequalities, the same conclusion holds if the sufficient condition in Theorem~\ref{thm:suff_cond} and the population no-harm condition hold. The reported plug-in boundary and finite-box calculations generally support inverse covariance in the ordered-treatment simulations and gasoline-demand application. The counterexamples in Sections~\ref{subsec:three_counterexample} and~\ref{subsec:no_harm_example} show why these diagnostics are still useful: inverse covariance is well motivated as a default, but it is not universally optimal when several inequalities interact.

\section{Monte Carlo simulations} \label{sec:simulations}
This section studies the quantitative importance of projection geometry and the behavior of the practical selector. We organize the simulations around three settings. The first is a deliberately artificial Gaussian design constructed to show how the boundary-risk/no-harm rule behaves when inverse covariance is and is not supported by the sufficient condition. The next two are economically motivated, namely an ordered-treatment-effects design with a common control group and a flexible parametric IV design based on \citet{chetverikov2017nonparametric}.

\subsection{Illustrative boundary-risk selector}
\label{subsec:sim_selector}

We begin with an analytically transparent Gaussian design that illustrates the feasible selector when the boundary-risk and no-harm diagnostics do and do not support inverse covariance. Let $k=10$ and observe
\[
X_i\stackrel{iid}{\sim}N(\theta_0,I_{10}),
\qquad i=1,\ldots,n,
\]
with $\hat\theta_{ur}=\bar X$ and coordinatewise restrictions
\[
\theta_j\geq0,
\qquad j=1,\ldots,10.
\]
We set $n=250$. We start with exact-boundary designs in which the first $m\in\{1,3,5\}$ inequalities bind. In these designs $\theta_{0,j}=0$ for $j\leq m$, while the remaining coordinates are fixed at $0.75$ in the interior, making their inequalities asymptotically irrelevant. The true covariance of $\sqrt n(\hat\theta_{ur}-\theta_0)$ is $I_{10}$, but the feasible procedure uses the sample covariance matrix in every replication.

The loss matrix is
\[
H=
\begin{pmatrix}
H_5&0\\
0&0.05I_5
\end{pmatrix},
\qquad
H_5=0.05I_5+\mathbf 1\mathbf 1'.
\]
The block $H_5$ places substantial weight on common movements of the first five coordinates, so risk depends strongly on how errors interact across locally relevant inequalities. Thus the loss and covariance geometries differ even though sampling covariance is proportional to the identity. The feasible selector uses the conservative slack screen with $\kappa=5$, the normalized path in \eqref{eq:path}, and the finite-box diagnostic with $M=8$. The results use 1,000 outer replications. In the designs reported below the screen retains the true locally relevant set in every replication with no false positives or false negatives.

Table~\ref{tab:selector_illustrative} summarizes the exact-boundary designs. Selection is performed separately in each replication on the grid $\{0,0.05,\ldots,1\}$, using a $0.004$ numerical allowance for the finite-box diagnostic and always admitting the analytically safe loss projection. Appendix~\ref{appendix:replication_details} explains the allowance. With $m=1$ and $m=3$, the sufficient condition in Theorem~\ref{thm:suff_cond} holds in the population, and the selector chooses $s=1$ in every replication. With $m=5$, the sufficient condition fails in the population and every replication, and normalized worst boundary risk exceeds one at $s=1$. The selector moves toward loss projection, choosing $\bar s_{BR}=0.722$ on average and reducing estimated worst boundary risk from $1.094$ to $0.933$.

\begin{table}[!htbp]
\centering
\caption{Illustrative boundary-risk selector}
\label{tab:selector_illustrative}
\footnotesize
\begin{tabular}{cccccc}
\toprule
$m$ & \makecell{Suff.\\condition} & $\bar s_{BR}$ & $\overline{\hat B(0)}$ & $\overline{\hat B(1)}$ & $\overline{\hat B(\hat s_{BR})}$\\
\midrule
1 & 1.00 & 1.000 & 0.994 & 0.903 & 0.903\\
3 & 1.00 & 1.000 & 0.995 & 0.922 & 0.922\\
5 & 0.00 & 0.722 & 0.995 & 1.094 & 0.933\\
\bottomrule
\end{tabular}
\begin{minipage}{0.92\linewidth}
\vspace{2mm}
{\footnotesize\begin{singlespace}
Notes: $n=250$ and 1,000 outer replications. \emph{Suff. condition} is the share satisfying the plug-in sufficient condition in Theorem~\ref{thm:suff_cond}. $\bar s_{BR}$ averages replication-specific selections. $\hat B$ is normalized by plug-in unrestricted risk, and the last column averages boundary risk at the selected values. Selection requires $\hat G_8(s)\leq1.004$ for $s>0$, with $\hat G_8(0)=1$ imposed analytically. Appendix~\ref{appendix:replication_details} specifies the numerical rule.
\end{singlespace}}
\end{minipage}
\end{table}

\begin{figure}[b!]
	\centering
	\caption{Illustrative selector when five inequalities are locally relevant}
	\label{fig:selector_m5}
	\includegraphics[width=\textwidth]{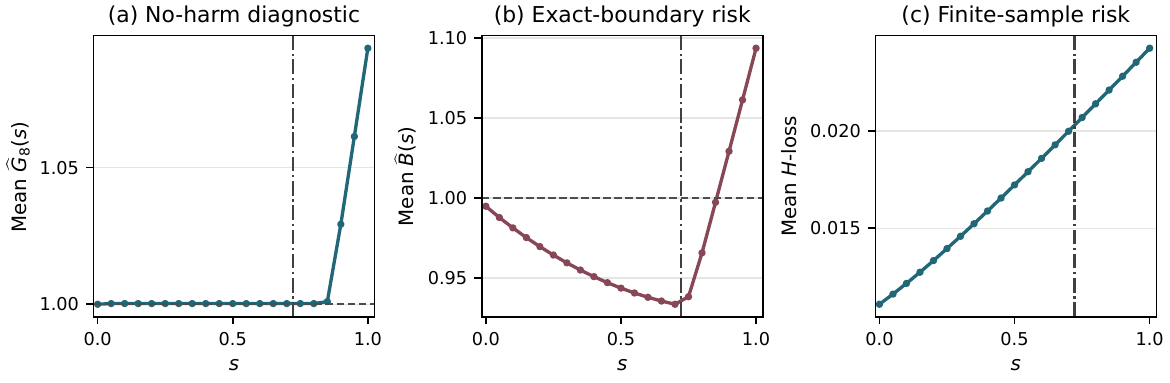}
	\begin{minipage}{0.92\linewidth}
		\vspace{1mm}
		{\footnotesize\begin{singlespace}
				Notes: Average finite-box diagnostic $\hat G_8(s)$, exact-boundary criterion $\hat B(s)$, and finite-sample $H$-risk estimated from the same 1,000 replications of the $n=250$, $m=5$ exact-boundary design. The vertical marker is the average replication-specific selection $\bar s_{BR}$, not the minimizer of the averaged curve. The horizontal dashed line is one, and $\hat G_8(0)=1$ is imposed analytically.
		\end{singlespace}}
	\end{minipage}
\end{figure}

Figure~\ref{fig:selector_m5} shows average boundary risk and the finite-box no-harm diagnostic along the path for the $m=5$ exact-boundary design. The key feature is that the boundary-risk curve has an interior minimum among the values admitted by the finite-box diagnostic. The third panel reports finite-sample loss at the particular data-generating process. This loss need not be minimized at $s_{BR}$ because the selector protects against all exact-boundary faces generated by the retained inequalities rather than optimizing risk at a known active face.

Mixed local-slack profiles likewise favor inverse covariance for $m=1,3$ and movement away from $s=1$ for $m=5$. Appendix~\ref{appendix:selector_details} reports these results and implementation details.

\FloatBarrier

\subsection{Ordered treatment effects}
\label{subsec:sim_treatment}

We next consider the ordered-treatment-effects design of Example~\ref{ex:treatment}, in which covariance across coordinates arises directly from an economically familiar setting. There is a control group $d=0$ and five ordered treatment intensities $d=1,\ldots,5$. Define the treatment effects
\[
\theta_d=E[Y(d)]-E[Y(0)],
\qquad d=1,\ldots,5,
\]
and impose ordered average treatment effects,
\[
0\leq\theta_1\leq\theta_2\leq\cdots\leq\theta_5.
\]
Writing $\theta=(\theta_1,\ldots,\theta_5)'$, the restrictions are $A\theta\geq0$ with
\[
A=
\begin{pmatrix}
1&0&0&0&0\\
-1&1&0&0&0\\
0&-1&1&0&0\\
0&0&-1&1&0\\
0&0&0&-1&1
\end{pmatrix}.
\]

Conditional on deterministic group counts, outcomes are independent Gaussian,
\[
Y_{ig}\sim N(m_g,\sigma_g^2),
\qquad m_0=0,\quad m_d=\theta_d,
\]
and the unrestricted estimator is $\hat\theta_{ur,d}=\bar Y_d-\bar Y_0$. Its scaled covariance matrix has entries
\[
\Sigma_{dd}=\frac{n\sigma_d^2}{n_d}+\frac{n\sigma_0^2}{n_0},
\qquad
\Sigma_{d\ell}=\frac{n\sigma_0^2}{n_0},\quad d\neq\ell.
\]
The common control therefore generates positive dependence across all five treatment-effect estimates. We use Euclidean loss, $H=I_5$, so the loss projection is ordinary Euclidean projection.

The baseline grid combines five shape configurations with three sampling designs and sample sizes $n\in\{100,400,1600\}$. The configurations are interior, one binding increment, three binding increments, and mixed local-slack profiles involving three or all five increments. Two additional stress designs have the first constraint binding and the remaining adjacent-difference restrictions eight standardized local units inside the feasible set. One has a small common control group, and the other has a noisy common control. Appendix~\ref{appendix:ordered_treatment_details} gives the full design grid.

Table~\ref{tab:treatment_selector_summary} summarizes the 51 design cells. The first five rows each aggregate nine baseline cells, corresponding to three sample sizes and three baseline sampling designs. The final row aggregates the six stress cells, corresponding to three sample sizes and two stress designs. Each cell uses 1,000 Monte Carlo replications, estimated covariance matrices, and retained inequalities selected by the conservative screen.

In every cell, the sufficient condition holds at the inverse-covariance endpoint and the average finite-box diagnostic is close to one, with maximum average excess about $0.00021$ across cells. Under the same replication-specific numerical rule as above, the selector chooses $s=1$ in all 51,000 replications. Appendix~\ref{appendix:replication_details} explains why the $0.004$ allowance accommodates the observed numerical integration error. The conservative screen retains all truly binding constraints in every replication. In the stress designs, it drops almost all inequalities with large local slack but keeps the binding inequality. Thus the reported calculations support inverse covariance, in contrast to the deliberately adverse Gaussian design.

\begin{table}[!htbp]
\centering
\caption{Boundary-risk selector in the ordered-treatment designs}
\label{tab:treatment_selector_summary}
\footnotesize
\setlength{\tabcolsep}{3pt}
\begin{tabular}{@{}lrrrrrrr@{}}
\toprule
Shape configuration & Cells & $\bar s_{BR}$ & $\overline{|\hat J|}$ & \makecell{Prob. binding\\retained} & $\overline{\hat B(\hat s_{BR})}$ & $L_{BR}/L_0$ & $L_{BR}/L_{ur}$\\
\midrule
Interior & 9 & 1.000 & 4.99--5.00 & -- & 0.940--0.961 & 0.930--0.997 & 0.448--0.920\\
One binding & 9 & 1.000 & 5.00 & 1.000 & 0.940--0.961 & 0.930--1.000 & 0.459--0.902\\
Three binding & 9 & 1.000 & 4.99--5.00 & 1.000 & 0.940--0.961 & 0.955--1.025 & 0.415--0.790\\
Three local mixed & 9 & 1.000 & 4.99--5.00 & 1.000 & 0.940--0.961 & 0.923--0.999 & 0.530--0.886\\
All local mixed & 9 & 1.000 & 4.99--5.00 & 1.000 & 0.940--0.961 & 0.920--0.988 & 0.630--0.729\\
Stress designs & 6 & 1.000 & 1.00--1.02 & 1.000 & 0.578--0.682 & 0.617--0.759 & 0.551--0.684\\
\bottomrule
\end{tabular}
\begin{minipage}{0.92\linewidth}
\vspace{2mm}
{\footnotesize\begin{singlespace}
Notes: Each cell uses 1,000 replications. Entries give ranges of cell averages. $\bar s_{BR}$ is the average selected path value and $\overline{|\hat J|}$ the average retained-set size. \emph{Prob. binding retained} is the share retaining all truly binding inequalities, which is not applicable in the interior row. $L_{BR}$, $L_0$, and $L_{ur}$ are mean Euclidean losses for the selector, loss projection, and unrestricted estimation. The plug-in sufficient condition holds throughout. Selection uses $\hat G_8(s)\leq1.004$ for $s>0$ and the analytical guarantee at $s=0$. Inverse covariance is selected in every replication.
\end{singlespace}}
\end{minipage}
\end{table}

The selected estimator lowers mean Euclidean loss relative to unrestricted estimation in every cell, with loss ratios from $0.415$ to $0.920$. Relative to loss projection, it has lower point-estimated risk in 46 of the 51 cells. The largest excess in the remaining cells is about $2.5\%$. The largest gains occur in the stress designs, where the common-control component dominates covariance and the selector reduces loss by roughly $24$--$38\%$ relative to $s=0$.

Appendix~\ref{appendix:ordered_treatment_details} reports the full cell-level table and a representative path for the $n=400$ balanced-homoskedastic design with standardized local slacks $(0,0.5,1,1.5,2)$.

\FloatBarrier

\subsection{Flexible IV design}
\label{subsec:sim_npiv}

We finally consider a flexible parametric IV design inspired by \citet{chetverikov2017nonparametric}. Consider
\[
Y=g(X)+\varepsilon,
\qquad E[\varepsilon\mid Z]=0,
\]
with $g$ approximated by the basis $(1,x,x^2)$ and instrument basis $(1,z,z^2)$. Complete details of the DGP are reported in Appendix~\ref{appendix:flexible_iv_details}. Because the model is just identified, changing the GMM weighting matrix leaves the unrestricted estimator unchanged while changing the geometry of the constrained projection.

Monotonicity of $g$ on $[0,1]$ implies
\[
\theta_2\geq0,
\qquad
\theta_2+2\theta_3\geq0,
\]
so the constraint matrix is
\[
A=\begin{pmatrix}0&1&0\\0&1&2\end{pmatrix}.
\]
Unlike the previous two simulations, this design does not require a numerical selector. There are only two maintained monotonicity inequalities, so Theorem~\ref{thm:two_const} implies that inverse-covariance projection solves Problem~\eqref{eq:population_choice}: it is boundary minimax and locally no-harm without a finite-box calculation. We use this simulation instead to compare the finite-sample consequences of alternative IV-induced projection geometries.

We compare 2SLS, efficient IV/inverse-covariance projection, identity-GMM-weight projection, loss projection, Euclidean projection, rearrangement, and the unrestricted estimator. The 2SLS entry projects using the geometry induced by the conventional 2SLS weighting matrix, and identity-GMM-weight projection uses the geometry induced by the identity GMM weighting matrix. Euclidean projection imposes the restrictions using $\Omega=I$. It is included only as an informal benchmark because this geometry has no direct interpretation as either efficient IV/inverse-covariance projection or loss projection in this example. Performance is measured by a scaled mean integrated squared error, which corresponds to the quadratic coefficient loss derived in Example~\ref{ex:fnc_iv}. Table~\ref{tab:npiv_main} summarizes the main results.

\begin{table}[!htbp]
\centering
\caption{Flexible IV projection geometries: MISE $\times 1000$}
\label{tab:npiv_main}
\small
\begin{tabular}{lrr}
\toprule
Estimator & $g(x)=x^2$ & $g(x)=0$\\
\midrule
2SLS & 24.86 & 14.46\\
Efficient IV / inverse covariance & 15.69 & 4.09\\
Rearrangement & 32.42 & 49.41\\
Identity GMM & 25.77 & 15.30\\
Loss projection & 26.69 & 17.75\\
Euclidean projection & 757.91 & 1041.56\\
Unrestricted & 47.05 & 49.42\\
\bottomrule
\end{tabular}
\begin{minipage}{0.92\linewidth}
\vspace{2mm}
{\footnotesize\begin{singlespace}
Notes: Entries are mean integrated squared errors multiplied by 1,000 and averaged over 1,000 Monte Carlo replications with $n=1{,}000$. Identity GMM uses the identity-GMM-weight projection geometry, and Euclidean projection is the informal $\Omega=I$ benchmark. Rearrangement is the procedure of \citet{CFG:09}, with its MISE evaluated using 10,000 auxiliary draws per replication.
\end{singlespace}}
\end{minipage}
\end{table}

When $g(x)=x^2$, the derivative binds at $x=0$, so the first monotonicity inequality binds. Efficient IV/inverse-covariance projection has scaled MISE $15.69$, compared with $24.86$ for 2SLS and $47.05$ for the unrestricted estimator. When $g(x)=0$, the derivative is zero throughout, so both endpoint inequalities bind and efficient IV/inverse-covariance projection has scaled MISE $4.09$. Thus the inverse-covariance geometry justified by Theorem~\ref{thm:two_const} also performs strongly in finite samples in this IV design.

Figure~\ref{fig:npiv_quadratic_sel} summarizes the graphical evidence by comparing efficient IV/inverse-covariance projection and 2SLS directly. The difference between efficient IV/inverse-covariance projection and 2SLS is substantial even though the unrestricted estimator is identical. Appendix~\ref{appendix:flexible_iv_details} reports additional implementation details and separate estimator panels for the quadratic and constant designs. Together, these figures illustrate the main point of the paper in an econometric IV setting: the restrictions define the feasible set, but the projection geometry determines how an unrestricted estimate is moved back to that set when it violates the restrictions.

\begin{figure}[!htbp]
\centering
\caption{Comparison of 2SLS and efficient IV/inverse-covariance projection in flexible IV}
\label{fig:npiv_quadratic_sel}
\includegraphics[width=0.92\textwidth]{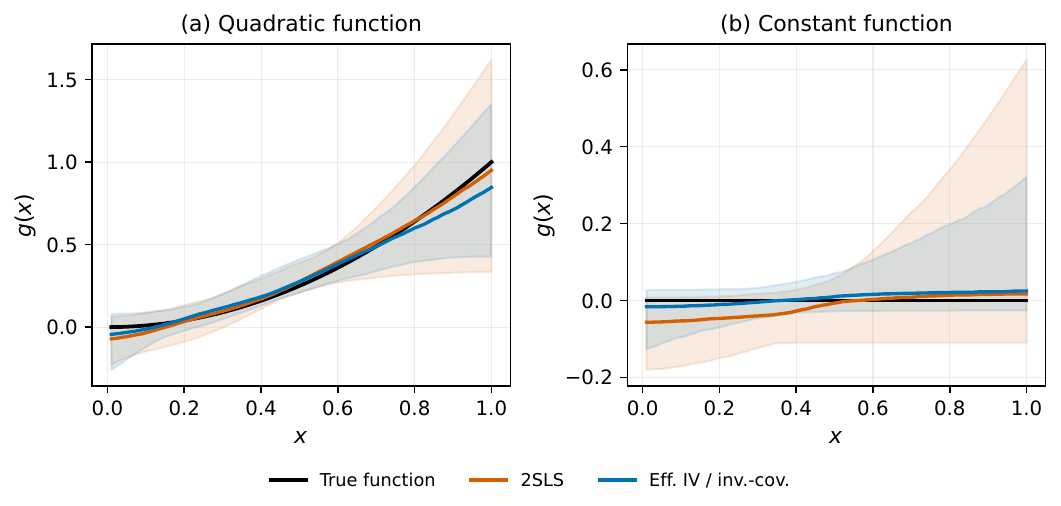}
\begin{minipage}{0.90\linewidth}
\vspace{1mm}
{\footnotesize\begin{singlespace}
Notes: Side-by-side comparison of 2SLS and efficient IV/inverse-covariance projection. Central curves are pointwise medians and bands span the pointwise 5th--95th percentiles across all 1,000 replications. The quadratic function has one truly binding monotonicity constraint and the constant function has two.
\end{singlespace}}
\end{minipage}
\end{figure}

Appendix~\ref{appendix:geometry_illustrations} reports additional two-dimensional illustrations that isolate the geometry behind these simulation results.

\FloatBarrier
\section{Application} \label{sec:application}
We apply the framework to the gasoline-demand application of \citet{blundell2012measuring}. The same application is used by \citet{chetverikov2017nonparametric} to compare unrestricted and shape-constrained IV estimators. Gasoline price is potentially endogenous. We maintain weakly decreasing uncompensated demand as an economically motivated restriction, rather than an unconditional implication of consumer theory. The application is therefore a natural setting for studying projection geometry.

We adapt the implementation to isolate projection geometry, using a just-identified spline IV specification with income interactions, a smaller set of controls, and a nonredundant monotonicity system for the displayed income range. Let
\[
Y=g(X_1,X_2)+\gamma'X_3+\varepsilon,
\]
where $Y$ is log gasoline consumption, $X_1$ is log price, $X_2$ is log household income, and $X_3$ collects additional controls. Price is instrumented using distance to a major oil platform, denoted by $Z$, under the conditional moment restriction $E[\varepsilon\mid Z,X_2,X_3]=0$.

The estimation sample contains 4,812 observations. We approximate the price and instrument components using quadratic spline bases with three internal knots at empirical quartiles. The regressor vector contains the six price spline terms, their interactions with income, and four controls. The instrument vector is formed analogously from the six distance spline terms, their interactions with income, and the same controls. The resulting 16-dimensional IV problem is just identified, so the unrestricted estimator is common across the GMM weighting matrices considered below. Changing the weighting matrix changes only the geometry used to impose monotonicity.

The controls $X_3$ are log household size, log number of drivers, log age of the household head, and the number of workers. We report curves at three log-income levels, $10.657$, $11.043$, and $11.408$, labeled low, medium, and high income. At each income, weakly decreasing demand is imposed through five derivative inequalities at the three internal knots and the two support endpoints, using one-sided derivatives at the endpoints. Since the derivative is piecewise linear in price, these inequalities enforce monotonicity throughout the price support. It is also affine in income, so the low- and high-income inequalities imply those at every income between them. All constrained curves in Figure~\ref{fig:gasoline} therefore impose the same 10 nonredundant restrictions jointly. The plotted curves fix the controls at sample means.

\begin{figure}[!htbp]
   \caption{Gasoline demand for different income groups}
   \label{fig:gasoline}
   \begin{center}
      \includegraphics[width=0.92\textwidth]{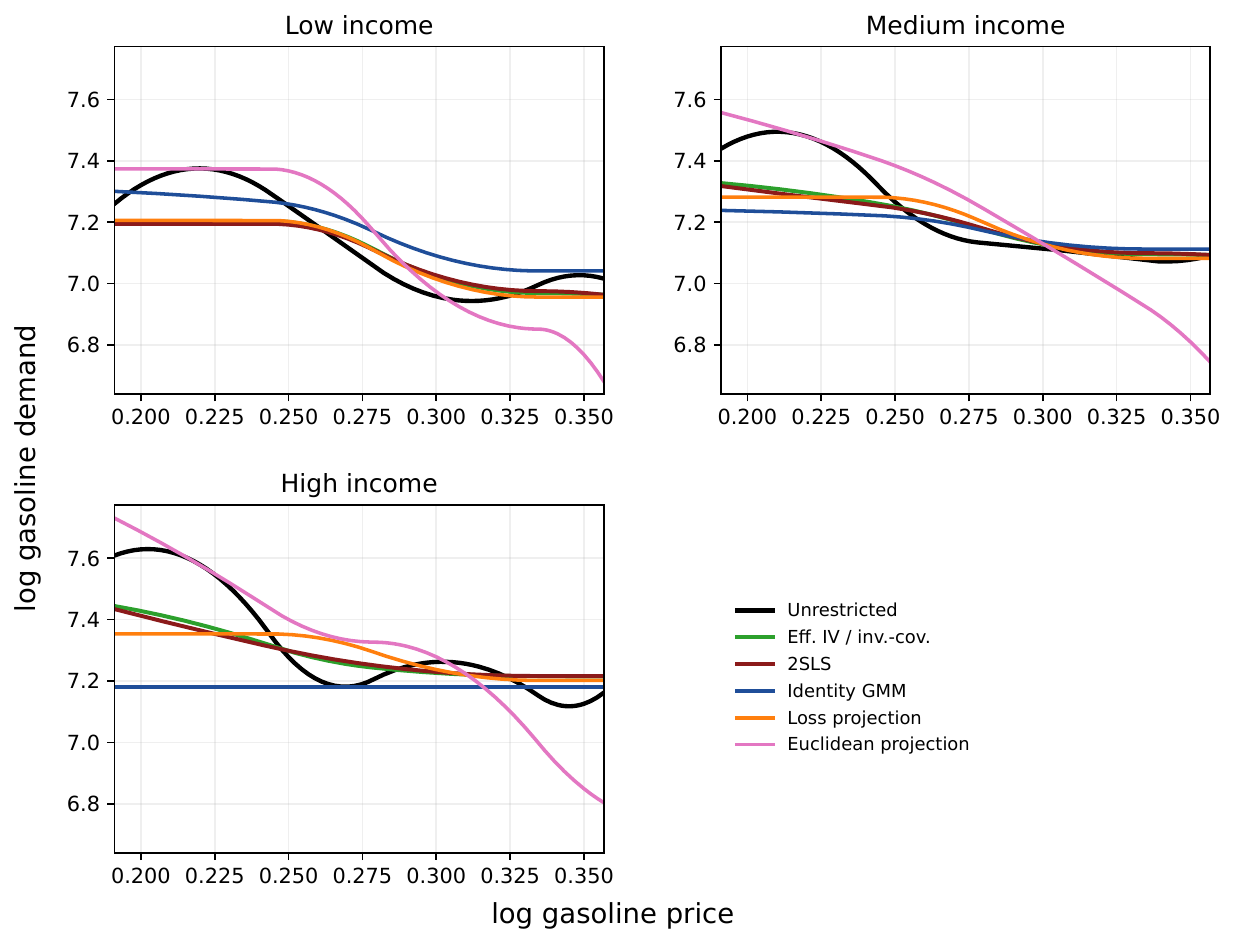}
   \end{center}
   {\footnotesize \begin{singlespace}
   Notes: Estimated gasoline-demand curves for three income groups. The black line is the unrestricted IV estimator. The constrained curves impose the 10 nonredundant low- and high-income monotonicity restrictions using efficient IV/inverse-covariance projection, 2SLS projection, identity-GMM-weight projection, loss projection, and Euclidean projection.
   \end{singlespace}}
\end{figure}

The projection-choice diagnostic uses the same 10 restrictions and a common loss criterion for the three income groups. For log income $y$, let $r_y(x,w)$ be the regressor vector mapping the full coefficient vector into fitted log demand at log price $x$ and controls $w$. Define
\[
H_y=\int r_y(x,w)r_y(x,w)'\,d\nu(x,w),
\]
where $\nu$ is the empirical distribution of log prices and controls in the estimation sample. We use
\[
H_{\mathrm{all}}=\frac{1}{3}\left(H_{\mathrm{low}}+H_{\mathrm{medium}}+H_{\mathrm{high}}\right)
\]
for both loss projection and the diagnostic. Quadratic loss under $H_{\mathrm{all}}$ is average squared prediction error in log demand, giving equal weight to the three income groups. It averages over observed prices and controls rather than fixing controls at their means as in the figure. This evaluates the fitted demand functions rather than the scale-dependent coefficient errors. The empirical $H_{\mathrm{all}}$ is positive definite in the full 16-dimensional parameterization, so loss projection is well defined.

The unrestricted curves have regions in which estimated demand increases with price. The constrained curves are all weakly decreasing, but imposing the same inequalities yields visibly different estimates. Efficient IV/inverse-covariance and 2SLS projection produce similar curves. Identity-GMM-weight projection is flatter, particularly at high income, while Euclidean projection produces a much steeper decline. Loss projection also differs from the IV-based geometries. These differences arise entirely from the projection step applied to the same unrestricted estimator.

We apply the plug-in diagnostic from Section~\ref{sec:selection}, using the heteroskedasticity-robust covariance estimate $\hat\Sigma$. The conservative screen with $\kappa=5$ retains all 10 inequalities. The restriction matrix has full row rank, so all $2^{10}-1=1023$ nonempty binding sets are attainable. At the inverse-covariance endpoint, we estimate risk for each exact-boundary configuration and numerically search finite slacks in $[0,8]^{|L|}$ for every nonempty retained subset $L$. The latter calculation uses paired excess losses relative to unrestricted estimation and includes the all-interior benchmark. Table~\ref{tab:gasoline_all_diagnostic} reports the normalized risks, and Appendix~\ref{appendix:app_additional} provides results by the number of binding inequalities.

\begin{table}[!htbp]
\centering
\caption{Endpoint projection-choice diagnostic in the application}
\label{tab:gasoline_all_diagnostic}
\footnotesize
\begin{tabular}{rrrrr}
\toprule
Retained & Faces & $\hat B_1(1)$ & $\hat B_{\geq2}(1)$ & $\hat G_8(1)$\\
\midrule
10/10 & 1023 & 0.892 & 0.866 & 1.0000\\
\bottomrule
\end{tabular}
\begin{minipage}{0.90\linewidth}
\vspace{2mm}
{\footnotesize\begin{singlespace}
Notes: The endpoint diagnostic uses efficient IV/inverse-covariance projection, the joint low/high-income restriction system, and loss $H_{\mathrm{all}}$. $\hat B_1(1)$ and $\hat B_{\geq2}(1)$ are the largest estimated exact-boundary risks with one and at least two binding inequalities, respectively. $\hat G_8(1)$ searches all 1023 nonempty retained subsets and standardized slacks in $[0,8]^{|L|}$, including the all-interior benchmark. The largest estimated nonempty-subset excess risk is zero at the reported numerical precision. All risk quantities are normalized by $\operatorname{tr}(H_{\mathrm{all}}\hat\Sigma)$. Calculations use 8,192 scrambled Sobol draws per subset.
\end{singlespace}}
\end{minipage}
\end{table}

Both diagnostics support the inverse-covariance endpoint. The largest multi-binding boundary risk is below the largest one-binding risk,
\[
\hat B_{\geq2}(1)=0.865799<0.892252=\hat B_1(1),
\]
so the estimated risks satisfy the sufficient condition in Theorem~\ref{thm:suff_cond}. The subset-maximized diagnostic gives $\hat G_8(1)=1.0000$, without a positive numerical allowance. The endpoint-first rule therefore retains efficient IV/inverse-covariance projection without searching the path. These plug-in results do not establish population no-harm beyond the finite boxes or exclude maxima missed by numerical search.

\section{Conclusion} \label{sec:summary}
This paper studies the choice of projection geometry for estimators subject to linear inequality constraints. Projection-based estimators arise naturally from constrained least squares, IV, GMM, and maximum likelihood, but the geometry used to impose the restrictions is often chosen implicitly. We show that this choice can matter for quadratic risk when the sampling covariance, loss geometry, and projection geometry are not aligned.

The results give a qualified but useful case for inverse-covariance projection. It is pointwise optimal at exact-boundary configurations where only one maintained inequality binds and boundary minimax over exact-boundary configurations with at most two binding inequalities. It is also locally no-harm when at most two inequalities have finite local slack. For larger systems, boundary minimaxity continues to hold whenever the worst boundary risk under inverse covariance is attained at a one-binding configuration. At the same time, examples with several constraints show that inverse-covariance projection is not always boundary minimax and need not always reduce risk relative to unrestricted estimation.

These limitations motivate the practical selector. We choose a projection geometry by minimizing exact-boundary risk subject to a local no-harm condition relative to the unrestricted estimator. When the maintained system has at most two inequalities, inverse covariance solves this population problem analytically. For larger systems, the reported plug-in boundary and finite-box diagnostics often support inverse covariance without further search. This occurs in all ordered-treatment replications and in the gasoline-demand application, while the flexible-IV simulation shows sizable finite-sample gains from the same geometry in a two-constraint design. The deliberately adverse Gaussian design shows that the selector can also move away from inverse covariance. Thus inverse covariance remains a natural prespecified geometry, while the diagnostics and selector provide safeguards when interactions among constraints favor another geometry.

\newpage

\begin{singlespace}
\bibliographystyle{chicago}
\bibliography{bibliography}

@article{chetverikov2017nonparametric,
  title={{Nonparametric Instrumental Variable Estimation Under Monotonicity}},
  author={Chetverikov, Denis and Wilhelm, Daniel},
  journal={Econometrica},
  volume={85},
  number={4},
  pages={1303--1320},
  year={2017},
  publisher={Wiley Online Library}
}

@unpublished{fang2014optimal,
  title={{Optimal Plug-in Estimators of Directionally Differentiable Functionals}},
  author={Fang, Zheng},
  year={2014},
  note = {Working paper}
}

@unpublished{cox2024,
  title={{A Simple and Adaptive Confidence Interval when Nuisance Parameters Satisfy an Inequality}},
  author={Cox, Gregory Fletcher},
  year={2026},
  note = {Accepted at the Journal of the American Statistical Association; accepted author version posted online August 21, 2026; arXiv:2409.09962}
}

@article{armstrong2018,
author = {Armstrong, Timothy B. and Kolesár, Michal},
title = {Optimal Inference in a Class of Regression Models},
journal = {Econometrica},
volume = {86},
number = {2},
pages = {655 -- 683},
year = {2018}
}

@article{donoho1994,
author = {David L. Donoho},
title = {{Statistical Estimation and Optimal Recovery}},
volume = {22},
journal = {The Annals of Statistics},
number = {1},
publisher = {Institute of Mathematical Statistics},
pages = {238 -- 270},
year = {1994}}

@unpublished{freybergerreeves2018,
  title={Inference under shape restrictions},
  author={Freyberger, Joachim and Reeves, Brandon},
  year={2018},
  note = {Working paper, SSRN 3011474},
  doi={10.2139/ssrn.3011474}
}

@article{chetverikov2018econometrics,
  title={{The Econometrics of Shape Restrictions}},
  author={Chetverikov, Denis and Santos, Andres and Shaikh, Azeem M},
  journal={Annual Review of Economics},
  volume={10},
  pages={31--63},
  year={2018},
  publisher={Annual Reviews}
}

@Book{zbMATH03942782,
 Author = {Eaton, Morris L.},
 Title = {Multivariate statistics. {A} vector space approach},
 FSeries = {Wiley Series in Probability and Mathematical Statistics},
 Series = {Wiley Ser. Probab. Math. Stat.},
 Year = {1983},
 Publisher = {John Wiley \& Sons, Hoboken, NJ},
 Language = {English},
 zbMATH = {3942782},
 Zbl = {0587.62097}
}

@article{blundell2012measuring,
  title={Measuring the price responsiveness of gasoline demand: Economic shape restrictions and nonparametric demand estimation},
  author={Blundell, Richard and Horowitz, Joel L and Parey, Matthias},
  journal={Quantitative Economics},
  volume={3},
  number={1},
  pages={29--51},
  year={2012},
  publisher={Wiley Online Library}
}

@article{chernozhukov2023constrained,
  title={Constrained conditional moment restriction models},
  author={Chernozhukov, Victor and Newey, Whitney K and Santos, Andres},
  journal={Econometrica},
  volume={91},
  number={2},
  pages={709--736},
  year={2023},
  publisher={Wiley Online Library}
}

@unpublished{astfalck2024posteriorprojectioninferenceconstrained,
      title={Posterior Projection for Inference in Constrained Spaces}, 
      author={Lachlan Astfalck and Deborshee Sen and Sayan Patra and Edward Cripps and David Dunson},
      year={2026},
      eprint={1812.05741},
      archivePrefix={arXiv},
      primaryClass={stat.ME},
      url={https://arxiv.org/abs/1812.05741},
      doi={10.48550/arXiv.1812.05741},
      note={arXiv:1812.05741v6; submitted to SIAM/ASA Journal on Uncertainty Quantification}
}

@article{johnson2018shape,
  title={Shape constraints in economics and operations research},
  author={Johnson, Andrew L and Jiang, Daniel R},
  journal={Statistical Science},
  volume={33},
  number={4},
  pages={527--546},
  year={2018},
  publisher={JSTOR}
}

@article{chernozhukov2010quantile,
  title={Quantile and probability curves without crossing},
  author={Chernozhukov, Victor and Fern{\'a}ndez-Val, Iv{\'a}n and Galichon, Alfred},
  journal={Econometrica},
  volume={78},
  number={3},
  pages={1093--1125},
  year={2010},
  publisher={Wiley Online Library}
}

@article{chen2021shape,
  title={Shape-enforcing operators for generic point and interval estimators of functions},
  author={Chen, Xi and Chernozhukov, Victor and Fernandez-Val, Ivan and Kostyshak, Scott and Luo, Ye},
  journal={Journal of Machine Learning Research},
  volume={22},
  number={220},
  pages={1--42},
  year={2021}
}

@article{ayer1955empirical,
  title={An empirical distribution function for sampling with incomplete information},
  author={Ayer, Miriam and Brunk, H Daniel and Ewing, George M and Reid, William T and Silverman, Edward},
  journal={The Annals of Mathematical Statistics},
  volume={26},
  number={4},
  pages={641--647},
  year={1955},
  publisher={JSTOR}
}

@article{Andrews:99,
 author = {Donald W. K. Andrews},
 journal = {Econometrica},
 number = {6},
 pages = {1341--1383},
 title = {Estimation When a Parameter is on a Boundary},
 volume = {67},
 year = {1999}}

@article{Brunk:55,
title = {Maximum Likelihood Estimates of Monotone Parameters},
author = {Brunk, H. D.},
year = {1955},
journal = {The Annals of Mathematical Statistics},
volume = {26},
number = {4},
pages = {607--616}}

@article{CLX:13,
author = {Cai, T. T. and Low, M. G. and Xia, Y.},
journal = "The Annals of Statistics",
number = "2",
pages = "722--750",
title = "Adaptive confidence intervals for regression functions under shape constraints",
volume = "41",
year = "2013"
}

@article{CGS:15,
 author = {Sabyasachi Chatterjee and Adityanand Guntuboyina and Bodhisattva Sen},
 number = {4},
 pages = {1774--1800},
title = {On risk bounds in isotonic and other shape restricted regression problems},
 volume = {43},
 year = {2015},
journal = {The Annals of Statistics}
}

@article{CFG:09,
 author = {V. Chernozhukov and I. Fern\'andez-Val and A. Galichon},
 journal = {Biometrika},
 number = {3},
 pages = {559-575},
 title = {Improving point and interval estimators of monotone functions by rearrangement},
 volume = {96},
 year = {2009}}

@article{Dumbgen:03,
author = {D\"umbgen, Lutz},
journal = "Bernoulli",
number = "3",
pages = "423--449",
title = "Optimal confidence bands for shape-restricted curves",
volume = "9",
year = "2003"
}

@article{FH:15,
title = "Identification and shape restrictions in nonparametric instrumental variables estimation ",
journal = "Journal of Econometrics ",
volume = "189",
number = "1",
pages = "41--53",
year = "2015",
author = "Joachim Freyberger and Joel L. Horowitz"
}

@article{Geyer:94,
author = "Geyer, Charles J.",
journal = "The Annals of Statistics",
number = "4",
pages = "1993--2010",
title = "On the Asymptotics of Constrained $M$-Estimation",
volume = "22",
year = "1994"
}

@article{Hildreth:54,
 author = {Clifford Hildreth},
 journal = {Journal of the American Statistical Association},
 number = {267},
 pages = {598--619},
 title = {Point Estimates of Ordinates of Concave Functions},
 volume = {49},
 year = {1954}}

@article{HL:17,
title = "Nonparametric estimation and inference under shape restrictions",
journal = "Journal of Econometrics",
volume = "201",
number = "1",
pages = "108 -- 126",
year = "2017",
author = "Joel L. Horowitz and Sokbae Lee"
}

@article{Wright:81,
author = "Wright, F. T.",
journal = "The Annals of Statistics",
number = "2",
pages = "443--448",
title = "The Asymptotic Behavior of Monotone Regression Estimates",
volume = "9",
year = "1981",
}

@article{Zhang:02,
author = "Zhang, Cun-Hui",
journal = "The Annals of Statistics",
number = "2",
pages = "528--555",
title = "Risk bounds in isotonic regression",
volume = "30",
year = "2002"
}

@article{Chatterjee2014,
  author = {Chatterjee, Sourav},
  title = {A New Perspective on Least Squares under Convex Constraint},
  journal = {The Annals of Statistics},
  volume = {42},
  number = {6},
  pages = {2340--2381},
  year = {2014}
}

@article{Bellec2018,
  author = {Bellec, Pierre C.},
  title = {Sharp Oracle Inequalities for Least Squares Estimators in Shape Restricted Regression},
  journal = {The Annals of Statistics},
  volume = {46},
  number = {2},
  pages = {745--780},
  year = {2018}
}

@article{GuntuboyinaSen2018,
  author = {Guntuboyina, Adityanand and Sen, Bodhisattva},
  title = {Nonparametric Shape-Restricted Regression},
  journal = {Statistical Science},
  volume = {33},
  number = {4},
  pages = {568--594},
  year = {2018}
}

@article{OymakHassibi2016,
  author = {Oymak, Samet and Hassibi, Babak},
  title = {Sharp {MSE} Bounds for Proximal Denoising},
  journal = {Foundations of Computational Mathematics},
  volume = {16},
  number = {4},
  pages = {965--1029},
  year = {2016}
}

@article{AmelunxenEtAl2014,
  author = {Amelunxen, Dennis and Lotz, Martin and McCoy, Michael B. and Tropp, Joel A.},
  title = {Living on the Edge: Phase Transitions in Convex Programs with Random Data},
  journal = {Information and Inference},
  volume = {3},
  number = {3},
  pages = {224--294},
  year = {2014}
}

@book{BarlowEtAl1972,
  author = {Barlow, Richard E. and Bartholomew, David J. and Bremner, James M. and Brunk, H. D.},
  title = {Statistical Inference under Order Restrictions: The Theory and Application of Isotonic Regression},
  publisher = {Wiley},
  address = {New York},
  year = {1972}
}

@book{RobertsonEtAl1988,
  author = {Robertson, Tim and Wright, F. T. and Dykstra, Richard L.},
  title = {Order Restricted Statistical Inference},
  publisher = {Wiley},
  address = {New York},
  year = {1988}
}

@article{OonoShinozaki2005,
  author = {Oono, Youhei and Shinozaki, Nobuo},
  title = {Estimation of Two Order Restricted Normal Means with Unknown and Possibly Unequal Variances},
  journal = {Journal of Statistical Planning and Inference},
  volume = {131},
  number = {2},
  pages = {349--363},
  year = {2005}
}

@article{Lee1988,
  author = {Lee, Chu-In Charles},
  title = {Quadratic Loss of Order Restricted Estimators for Treatment Means with a Control},
  journal = {The Annals of Statistics},
  volume = {16},
  number = {2},
  pages = {751--758},
  year = {1988},
  doi = {10.1214/aos/1176350833}
}

@article{GargMisra2025,
  author = {Garg, Naresh and Misra, Neeraj},
  title = {Some Unified Results on Isotonic Regression Estimators of Order Restricted Parameters of a General Bivariate Location/Scale Model},
  journal = {Metrika},
  volume = {88},
  number = {8},
  pages = {1115--1149},
  year = {2025},
  doi = {10.1007/s00184-024-00978-w}
}

@article{CaiLow2004Minimax,
  author = {Cai, T. Tony and Low, Mark G.},
  title = {Minimax Estimation of Linear Functionals over Nonconvex Parameter Spaces},
  journal = {The Annals of Statistics},
  volume = {32},
  number = {2},
  pages = {552--576},
  year = {2004}
}

@article{CaiLow2004Adaptation,
  author = {Cai, T. Tony and Low, Mark G.},
  title = {An Adaptation Theory for Nonparametric Confidence Intervals},
  journal = {The Annals of Statistics},
  volume = {32},
  number = {5},
  pages = {1805--1840},
  year = {2004}
}

@article{Armstrong2015,
  author = {Armstrong, Timothy B.},
  title = {Adaptive Testing on a Regression Function at a Point},
  journal = {The Annals of Statistics},
  volume = {43},
  number = {5},
  pages = {2086--2101},
  year = {2015}
}

@unpublished{KwonKwon2020,
  author = {Kwon, Koohyun and Kwon, Soonwoo},
  title = {Inference in Regression Discontinuity Designs under Monotonicity},
  note = {arXiv:2011.14216},
  year = {2020}
}

@article{FangSantos2019,
  author = {Fang, Zheng and Santos, Andres},
  title = {Inference on Directionally Differentiable Functions},
  journal = {The Review of Economic Studies},
  volume = {86},
  number = {1},
  pages = {377--412},
  year = {2019}
}

@book{SilvapulleSen2005,
  author = {Silvapulle, Mervyn J. and Sen, Pranab K.},
  title = {Constrained Statistical Inference: Inequality, Order, and Shape Restrictions},
  publisher = {Wiley},
  address = {Hoboken, NJ},
  year = {2005}
}

@book{vanEeden2006,
  author = {van Eeden, Constance},
  title = {Restricted Parameter Space Estimation Problems: Admissibility and Minimaxity Properties},
  publisher = {Springer},
  address = {New York},
  year = {2006}
}

@article{MarchandStrawderman2004,
  author = {Marchand, {\'E}ric and Strawderman, William E.},
  title = {Estimation in Restricted Parameter Spaces: A Review},
  journal = {Institute of Mathematical Statistics Lecture Notes -- Monograph Series},
  volume = {45},
  pages = {1--24},
  year = {2004}
}

@article{TsukumaKubokawa2008,
  author = {Tsukuma, Hisayuki and Kubokawa, Tatsuya},
  title = {Stein's Phenomenon in Estimation of Means Restricted to a Polyhedral Convex Cone},
  journal = {Journal of Multivariate Analysis},
  volume = {99},
  number = {1},
  pages = {141--164},
  year = {2008}
}

@article{RuedaSalvador1995,
  author = {Rueda, Cristina and Salvador, Bonifacio},
  title = {Reduction of Risk Using Restricted Estimators},
  journal = {Communications in Statistics -- Theory and Methods},
  volume = {24},
  number = {4},
  pages = {1011--1022},
  year = {1995},
  doi = {10.1080/03610929508831536}
}

@article{ShinozakiChang1999,
  author = {Shinozaki, Nobuo and Chang, Yuan-Tsung},
  title = {A Comparison of Maximum Likelihood and Best Unbiased Estimators in the Estimation of Linear Combinations of Positive Normal Means},
  journal = {Statistics \& Decisions},
  volume = {17},
  number = {2},
  pages = {125--136},
  year = {1999},
  doi = {10.1524/strm.1999.17.2.125}
}

@book{Theil1971,
  author = {Theil, Henri},
  title = {Principles of Econometrics},
  publisher = {Wiley},
  address = {New York},
  year = {1971}
}
\end{singlespace}

\newpage

\appendix

\section{Auxiliary results}
\label{appendix:proj_details}

This section collects auxiliary results used in the proofs. We first record two normalizations used throughout the main text. Only the symmetric part of a quadratic form matters because, for any square matrix $M$,
\[
x'Mx=x'\left(\frac{M+M'}{2}\right)x.
\]
We therefore restrict attention to symmetric positive-definite projection metrics.

Second, general quadratic loss can be transformed to Euclidean loss. Let $H\succ0$ and set $\tilde u=H^{1/2}u$. Then
\begin{align*}
R_H(\Omega,c)
&=E\left[\left\|H^{1/2}\mathcal V_r(\Omega,c)\right\|^2\right]\\
&=E\left[
\left\|
\argmin_{\tilde u:\,AH^{-1/2}\tilde u\geq c}
\|\tilde u-H^{1/2}\mathcal V_{ur}\|_{H^{-1/2}\Omega H^{-1/2}}^2
\right\|^2
\right].
\end{align*}
Thus the transformed problem has Euclidean loss, covariance
\[
\tilde\Sigma=H^{1/2}\Sigma H^{1/2},
\]
constraint matrix $\tilde A=AH^{-1/2}$, and projection metric
\[
\tilde\Omega=H^{-1/2}\Omega H^{-1/2}.
\]
In particular, $\tilde\Omega=\tilde\Sigma^{-1}$ if and only if $\Omega=\Sigma^{-1}$. Hence results proved after normalizing $H=I$ translate directly back to the original coordinates.

The first lemma gives a closed-form expression for equality projections.

\begin{lemma}
    \label{lem:general_proj}
    Let $A \in \R^{d \times k}$ be a matrix with $\operatorname{rank}(A) = d \leq k$, let $b \in \R^d$, let $\Omega \in \R^{k \times k} $ be a symmetric, positive definite matrix, and let $x \in \mathbb{R}^k$. Then 
   \begin{align*}
\argmin_{p \in  \R^k: A p = b  } \lVert p - x\rVert_{\Omega}^2 = x - \Omega^{-1} A' (A \Omega^{-1} A' )^{-1}(A x - b).
\end{align*}
\end{lemma}
\begin{proof}
    We incorporate the equality restriction $Ap = b$ by forming the Lagrangian
\[
    (p-x)'\Omega(p-x)+\lambda'(Ap-b),
\]
    with Lagrange multiplier $\lambda=(\lambda_1,\ldots,\lambda_d)'$.
    The first-order conditions are
\[
  2 \Omega( p - x) + A' \lambda  = 0  \qquad 
 \text{and}  \qquad  Ap - b = 0.
\]
Premultiplying the first condition by $A\Omega^{-1}$ and using $Ap=b$ gives
\[
    2(b-Ax)+A\Omega^{-1}A'\lambda=0,
\]
so
\[
    \lambda=2(A\Omega^{-1}A')^{-1}(Ax-b).
\]
Substituting this expression into the first-order condition yields
\[
p=x-\Omega^{-1} A' (A \Omega^{-1} A')^{-1}(A x-b).
\]
\end{proof}

The next lemma characterizes a conditional normal distribution.

\begin{lemma}
	 \label{lem:cond_normal}
	Let $Z \in \R^k$ satisfy $Z \sim \mathcal{N}(\mu, \Gamma)$, where $\Gamma$ has full rank. Let $A \in \R^{d \times k}$ with $\operatorname{rank}(A) = d \leq k$, and let $h \in \R^d$. Then 
    \[
   Z \mid A   Z  = h  \sim \mathcal{N}(\mu +  \Gamma A' (A \Gamma A')^{-1}(h-A\mu), \Gamma -  \Gamma A' (A \Gamma A')^{-1} A \Gamma)
    \] 
    and
    
\[
    (Z-\mu)-\Gamma A'(A\Gamma A')^{-1}A(Z-\mu)\mid AZ=h
    \sim
    \mathcal N(0,\Gamma-\Gamma A'(A\Gamma A')^{-1}A\Gamma).
    \]
\end{lemma}

\begin{proof}
	Consider the (degenerate) multivariate normal distribution 
	\[
		\left( \begin{array}{c} Z\\ A Z \end{array} \right) \sim \mathcal{N}\left( \left( \begin{array}{c}  \mu \\ A\mu \end{array} \right),\left( \begin{matrix} \Gamma& \Gamma A' \\ A\Gamma & A \Gamma A' \end{matrix} \right) \right).
	\]
	The standard conditional normal formula \citep{zbMATH03942782} gives the stated conditional mean and covariance matrix. The second part follows because $AZ$ is constant conditional on $AZ=h$.
\end{proof}

The following corollary follows immediately from Lemmas~\ref{lem:general_proj} and \ref{lem:cond_normal}.

\begin{corollary}
\label{cor:cond_dist_proj}
   Let $Z \in \R^k$ satisfy $Z \sim \mathcal{N}(0, \Gamma)$, where $\Gamma$ has full rank. Let $A \in \R^{d \times k}$ with $\operatorname{rank}(A) = d \leq k$, and let $h \in \R^d$. Then
   
\[
   \argmin_{p\in\R^k: Ap=0}\lVert p-Z\rVert_{\Gamma^{-1}}^2
   \mid AZ=h
   \sim
   \mathcal N(0,\Gamma-\Gamma A'(A\Gamma A')^{-1}A\Gamma).
   \]
\end{corollary}

The next corollary shows that the conditional covariance matrix of a normal random vector and its projection coincide after conditioning on the scalar constraint score.  

  \begin{corollary} 
    \label{cor:var_preserve}
    Let $Z \in \R^k$ satisfy $Z\sim \mathcal{N}(\mu, \Gamma)$, where $\Gamma$ has full rank. Let $a \in \mathbb{R}^{k \times 1} \backslash \{0\}$, let $b \in \R$, let $h \in \R$, and let $\Omega \in \R^{k \times k} $ be a symmetric, positive definite matrix. Define
    \[
    P_\Omega(Z)
    =
    \argmin_{p \in  \R^k: a' p \geq b  } \lVert p - Z\rVert_{\Omega}^2 .
    \]
    Then
    
\[
    Z-E[Z\mid a'Z=h]\mid a'Z=h
    \sim
    \mathcal N(0,\Gamma-\Gamma a(a'\Gamma a)^{-1}a'\Gamma)
    \]
    and
    
\[
    P_\Omega(Z)-E[P_\Omega(Z)\mid a'Z=h]\mid a'Z=h
    \sim
    \mathcal N(0,\Gamma-\Gamma a(a'\Gamma a)^{-1}a'\Gamma).
    \]
 \end{corollary}

\begin{proof}
The first part follows immediately from Lemma \ref{lem:cond_normal}. Moreover, if $h\geq b$, the second part holds since $Z$ is equal to its projection conditional on $a'Z = h$.

Lemma \ref{lem:general_proj} implies that for $Z$ with $a'Z < b$,

\[
\argmin_{p\in\R^k: a'p\geq b}\lVert p-Z\rVert_{\Omega}^2
=
Z-\Omega^{-1}a(a'\Omega^{-1}a)^{-1}(a'Z-b).
\]
Conditional on $a'Z = h$, the correction term is nonrandom. Centering therefore removes this term, and the conditional distribution of the centered projection is the same as the conditional distribution of the centered normal vector in the first part.
\end{proof}

The next lemma records a simple redundancy property that is useful in the
two-binding proof. For vector inequalities, all comparisons below are
understood componentwise.

\begin{lemma}
\label{lem:subset_constraints}
Let $c\in\{0,-\infty\}^d$ and define
\[
F_c=\{p\in\mathbb R^k:Ap\geq c\},
\qquad
F_0=\{p\in\mathbb R^k:Ap\geq0\}.
\]
Let $\bar A$ collect the rows of $A$ corresponding to coordinates for which
$c_j=-\infty$. For a symmetric positive-definite $\Omega$, let $p_c$ and
$p_0$ denote the $\Omega$-metric projections of $x$ onto $F_c$ and $F_0$,
respectively. If either $\bar A p_c>0$ or $\bar A p_0>0$, then $p_c=p_0$.
Consequently, $\bar A p_c>0$ if and only if $\bar A p_0>0$.
\end{lemma}

\begin{proof}
Because $c_j\leq0$ for every $j$, $F_0\subseteq F_c$. Suppose first that
$\bar A p_c>0$. The rows not contained in $\bar A$ have $c_j=0$, so
$p_c\in F_0$. Since $p_c$ minimizes the criterion over the larger set $F_c$
and is itself feasible for $F_0$, it must also be the unique minimizer over
$F_0$. Hence $p_c=p_0$.

Conversely, suppose $\bar A p_0>0$ and assume for contradiction that
$p_c\neq p_0$. Since $F_0\subseteq F_c$ and the objective is strictly
convex,
\[
\|p_c-x\|_\Omega^2<\|p_0-x\|_\Omega^2.
\]
For $t\in(0,1)$ define $p(t)=(1-t)p_0+tp_c$. The constraints whose
coordinates in $c$ equal zero are satisfied by both endpoints. The remaining
rows are strictly satisfied at $p_0$, so for all sufficiently small $t>0$,
$\bar A p(t)>0$ as well. Thus $p(t)\in F_0$. Strict convexity gives
\[
\|p(t)-x\|_\Omega^2
<
(1-t)\|p_0-x\|_\Omega^2+t\|p_c-x\|_\Omega^2
<
\|p_0-x\|_\Omega^2,
\]
contradicting optimality of $p_0$ over $F_0$. Thus $p_c=p_0$.
The final equivalence is immediate.
\end{proof}

\section{Proofs of main results}
\label{appendix:proofs}

\begin{proof}[Proof of Lemma \ref{lemma:minimax_entire}]
For the all-interior configuration $c=(-\infty,\ldots,-\infty)$, the feasible set in the local experiment is $\mathbb R^k$, so $\mathcal V_r(\Omega,c)=\mathcal V_{ur}$ for every $\Omega$. Hence
\[
\sup_{c\in C}R_H(\Omega,c)\geq R_{ur}
\]
for every projection geometry. On the other hand, $\mathcal V_r(H,c)$ is the $H$-metric projection of $\mathcal V_{ur}$ onto the closed convex set $U(c)$, which contains the origin for every $c\in C$. The projection inequality therefore gives
\[
\|\mathcal V_r(H,c)\|_H\leq\|\mathcal V_{ur}\|_H
\]
for every realization and every $c\in C$. Taking expectations yields $R_H(H,c)\leq R_{ur}$ uniformly over $C$. Thus the minimax value is $R_{ur}$. The final statement follows immediately for any other geometry whose risk is bounded by $R_{ur}$ throughout $C$.
\end{proof}

\begin{proof}[Proof of Lemma \ref{lemma:minimax_local}]
Write
\[
d=\frac{\Omega^{-1}a}{a'\Omega^{-1}a}
\]
for the projection direction generated by $\Omega$, so that $a'd=1$. Conversely, every vector $d$ satisfying $a'd=1$ can be generated by some symmetric positive-definite projection metric. Let
\[
P_a=I-\frac{aa'}{\|a\|^2}.
\]
For any $\tau>0$, set $K_d=dd'+\tau P_a$. Then $K_da=d$, and for any $x\neq0$,
\[
x'K_dx=(d'x)^2+\tau\|P_ax\|^2>0.
\]
Thus $K_d$ is symmetric positive definite and can be taken as $\Omega^{-1}$. We can therefore formulate the single-halfspace minimax problem directly over $d$ subject to $a'd=1$.

Let $\phi$ and $\Phi$ denote the standard normal density and distribution functions. The normalization $a'\Sigma a=1$ implies
\[
X=a'\mathcal V_{ur}\sim N(0,1).
\]
Let $v=\Sigma a$. Conditional normality gives
\[
\mathcal V_{ur}\mid X=x
\sim
N\left(vx,\,K\right),
\qquad
K=\Sigma-vv'.
\]
For local slack $c=-\delta$, the projection changes a draw only when $X<-\delta$. On that event,
\[
\mathcal V_r=\mathcal V_{ur}-d(X+\delta).
\]
Consequently,
\begin{align}
R_H(d,-\delta)
=&\ \operatorname{tr}(HK)
+\int_{-\infty}^{-\delta}
\{vx-d(x+\delta)\}'H\{vx-d(x+\delta)\}\phi(x)\,dx \notag\\
&+\int_{-\delta}^{\infty}x^2v'Hv\,\phi(x)\,dx.
\label{eq:finite_band_risk}
\end{align}

For part 1, the two integrands in \eqref{eq:finite_band_risk} agree at $x=-\delta$, so the boundary terms cancel when differentiating with respect to $\delta$. Hence
\begin{equation}
\frac{\partial R_H(d,-\delta)}{\partial\delta}
=
2\left[
\{d'Hv-d'Hd\}\phi(\delta)
+d'Hd\,\delta\Phi(-\delta)
\right].
\label{eq:finite_band_derivative}
\end{equation}
If $d'Hd-d'Hv\leq0$, this derivative is nonnegative for every $\delta\geq0$. Otherwise its sign is the sign of
\[
\frac{\delta\Phi(-\delta)}{\phi(\delta)}
-
\frac{d'Hd-d'Hv}{d'Hd}.
\]
Let $h(\delta)=\delta\Phi(-\delta)/\phi(\delta)$. For $\delta>0$,
\[
h'(\delta)
=
(1+\delta^2)\frac{\Phi(-\delta)}{\phi(\delta)}-\delta.
\]
The standard Mills lower bound $\Phi(-\delta)/\phi(\delta)>\delta/(1+\delta^2)$ implies $h'(\delta)>0$. Thus the derivative in \eqref{eq:finite_band_derivative} can change sign at most once, and if it changes sign it changes from negative to positive. Therefore $R_H(d,-\delta)$ has no interior maximum on $[0,\bar\delta]$, proving part 1.

We next derive the pointwise optimum, which also proves part 4. The terms in \eqref{eq:finite_band_risk} that depend on $d$ are
\[
A(\delta)d'Hd-2B(\delta)d'Hv,
\]
where
\begin{align*}
A(\delta)
&=\int_{-\infty}^{-\delta}(x+\delta)^2\phi(x)\,dx
=(1+\delta^2)\Phi(-\delta)-\delta\phi(\delta),\\
B(\delta)
&=\int_{-\infty}^{-\delta}x(x+\delta)\phi(x)\,dx
=\Phi(-\delta).
\end{align*}
Minimizing this strictly convex quadratic subject to $a'd=1$ gives
\begin{equation}
d_{pt}(\delta)
=
\frac{B(\delta)}{A(\delta)}\Sigma a
+
\left\{1-\frac{B(\delta)}{A(\delta)}\right\}
\frac{H^{-1}a}{a'H^{-1}a}.
\label{eq:pointwise_direction_H}
\end{equation}
Define
\[
w_1(\delta)=2\Phi(-\delta),
\qquad
w_2(\delta)=2\delta\{\delta\Phi(-\delta)-\phi(\delta)\}.
\]
Then $w_1(0)=1$, $w_2(0)=0$, and
\[
w_1(\delta)+w_2(\delta)=2A(\delta)>0.
\]
Therefore the matrix
\begin{equation}
\Omega^{-1}(\delta)
=
\frac{w_1(\delta)}{w_1(\delta)+w_2(\delta)}\Sigma
+
\frac{w_2(\delta)}{w_1(\delta)+w_2(\delta)}
\frac{H^{-1}}{a'H^{-1}a}
\label{eq:pointwise_metric_H}
\end{equation}
satisfies $a'\Omega^{-1}(\delta)a=1$ and generates the direction in \eqref{eq:pointwise_direction_H}. At $\delta=0$, \eqref{eq:pointwise_metric_H} equals $\Sigma$. Since positive definiteness is an open condition and the weights are continuous, $\Omega^{-1}(\delta)$ remains positive definite for all sufficiently small $\delta$. This proves part 4. Notice also that for $\delta>0$,
\[
A(\delta)-B(\delta)
=\delta\{\delta\Phi(-\delta)-\phi(\delta)\}<0,
\]
so the coefficient on the normalized $H^{-1}$ term is generally negative.

We now prove part 2. By part 1, the minimax problem is equivalent to minimizing
\[
M(d)=\max\{R_H(d,0),R_H(d,-\bar\delta)\}
\]
over $a'd=1$. At $\delta=0$, completing squares gives
\begin{equation}
R_H(d,0)
=r_0+\frac12(d-v)'H(d-v)
\label{eq:boundary_complete_square}
\end{equation}
for a constant $r_0$ independent of $d$. Hence $M(d)$ is continuous and coercive on the affine space $\{d:a'd=1\}$ and has a minimizer $d^*$. Suppose, contrary to part 2, that the exact boundary is the unique worst-case point at $d^*$, so that
\[
R_H(d^*,0)>R_H(d^*,-\bar\delta).
\]
The candidate $d^*$ cannot equal $v$, because \eqref{eq:finite_band_derivative} evaluated at $d=v$ gives
\[
\frac{\partial R_H(v,-\delta)}{\partial\delta}
=2v'Hv\,\delta\Phi(-\delta)>0
\qquad(\delta>0),
\]
and therefore $R_H(v,-\bar\delta)>R_H(v,0)$. Since $d^*\neq v$, moving a sufficiently small distance from $d^*$ toward $v$ strictly lowers \eqref{eq:boundary_complete_square}. By continuity, the exact boundary remains the worst-case endpoint for a sufficiently small move, so $M(d)$ falls, contradicting minimaxity. Thus a minimizer can be chosen for which $\delta=\bar\delta$ is a worst-case point, proving part 2.

Finally, consider inverse-covariance projection. Under the normalization $a'\Sigma a=1$, its projection direction is $v=\Sigma a$. As just shown,
\[
R_H(v,-\bar\delta)>R_H(v,0)
\qquad\text{for every }\bar\delta>0.
\]
If $\Sigma a$ is not proportional to $H^{-1}a$, then the two normalized directions $v$ and $H^{-1}a/(a'H^{-1}a)$ differ. Since $B(\bar\delta)/A(\bar\delta)>1$ for $\bar\delta>0$, \eqref{eq:pointwise_direction_H} implies $d_{pt}(\bar\delta)\neq v$. The outer-endpoint risk is a strictly convex quadratic in $d$ centered at $d_{pt}(\bar\delta)$, whereas the boundary risk is minimized at $v$. Moving a sufficiently small distance from $v$ toward $d_{pt}(\bar\delta)$ therefore lowers $R_H(d,-\bar\delta)$ to first order and raises $R_H(d,0)$ only to second order. Because the outer endpoint is strictly worse at $v$, it remains the worst endpoint for a sufficiently small move, and the worst-case risk falls. Thus inverse-covariance projection is not minimax for any $\bar\delta>0$, proving part 3.
\end{proof}

\begin{proof}[Proof of Theorem \ref{thm:one_const}]
Fix $c\in\bar C_1$ and let $j$ denote its unique binding coordinate. All other coordinates of $c$ equal $-\infty$, so the local problem reduces to the halfspace $a_j'u\geq0$. For brevity write $a=a_j$ and
\[
\sigma_a^2=a'\Sigma a,\qquad
d^*=\frac{\Sigma a}{\sigma_a^2},\qquad
d=d_{\Omega,j}.
\]
Conditioning on $X=a'\mathcal V_{ur}$, the conditional covariance of
$\mathcal V_{ur}$ does not depend on the projection geometry and the
conditional mean is $d^*X$. At the exact boundary, the two projected
conditional means differ only on $X<0$, where the difference equals
$(d^*-d)X$. Hence
\[
R_H(\Omega,c)-R_H(\Sigma^{-1},c)
=
E[X^2 1\{X<0\}]\,(d-d^*)'H(d-d^*)
=
\frac{a'\Sigma a}{2}(d-d^*)'H(d-d^*).
\]
Thus every positive-definite metric that generates $d^*$ attains the same
risk at this one-binding configuration, and $\Omega=\Sigma^{-1}$ generates
$d^*$ for every row $a_j$. This proves Part~1. Since Part~1 holds pointwise
for every $c\in\bar C_1$, it also implies
\[
\Sigma^{-1}\in\argmin_{\Omega\succ0}\max_{c\in\bar C_1}R_H(\Omega,c).
\]

For Part~2, fix $c\in C$ with at most one finite coordinate. If all coordinates
are $-\infty$, the restricted and unrestricted limits coincide and the claim is
immediate. Otherwise let $j$ be the unique finite coordinate, suppress the
asymptotically irrelevant inequalities, write $a=a_j$, and denote the scalar
finite slack by $\gamma=c_j\leq0$. By the transformation in
Appendix~\ref{appendix:proj_details}, it is enough to prove the no-harm claim
under Euclidean loss. Let
\[
X=a'\mathcal V_{ur},
\qquad
s^2=a'\Sigma a,
\qquad
d_\Omega=\frac{\Omega^{-1}a}{a'\Omega^{-1}a}.
\]
Conditional normality gives
\[
\mathcal V_{ur}\mid X=x
\sim
N\left(\frac{\Sigma a}{s^2}x,\,
\Sigma-\frac{\Sigma aa'\Sigma}{s^2}\right).
\]
If $x<\gamma$, the projected limit is
\[
\mathcal V_r(\Omega,c)
=
\mathcal V_{ur}-d_\Omega(x-\gamma);
\]
if $x\geq\gamma$, it equals $\mathcal V_{ur}$. The conditional covariance on
the violation region therefore does not depend on $\Omega$, while the
conditional mean is
\[
\frac{\Sigma a}{s^2}x-d_\Omega(x-\gamma).
\]
For the no-harm comparison, set $\Omega=\Sigma^{-1}$. Conditional on a
violation $X=x<\gamma$, the restricted conditional mean becomes
\[
\frac{\Sigma a}{s^2}\gamma,
\]
whereas the unrestricted conditional mean is
\[
\frac{\Sigma a}{s^2}x.
\]
The conditional covariance is the same for the two estimators. Since
$x<\gamma\leq0$ implies $x^2\geq\gamma^2$, the restricted conditional squared
bias is weakly smaller on every violating hyperplane. On the nonviolation
region the two estimators coincide. Integrating over $X$ yields
\[
E\!\left[\|\mathcal V_r(\Sigma^{-1},c)\|^2\right]\leq R_{ur}.
\]
Transforming back from Euclidean loss establishes Part~2 for general $H$.
\end{proof}

\begin{proof}[Proof of Theorem \ref{thm:two_const}]
We first prove Part~1. By the transformation in Appendix~\ref{appendix:proj_details}, it is enough
to prove the boundary-minimax result under Euclidean loss. Consider an attainable configuration
with two binding inequalities and denote their normals by $a_1$ and $a_2$.
All other inequalities are asymptotically irrelevant for this configuration.
Under the full-dimensionality convention in Section~\ref{subsec:setup} and
irredundancy of the maintained inequalities, two distinct simultaneously
active facets cannot have proportional normals; otherwise one inequality would
be redundant or the pair would encode a maintained equality. Maintained
equalities have already been absorbed into the parameterization. We therefore
consider the linearly independent case.

Suppressing the asymptotically irrelevant coordinates, for brevity write
\[
V^{12}=\mathcal V_r(\Sigma^{-1},(0,0)),\qquad
V^1=\mathcal V_r(\Sigma^{-1},(0,-\infty)),\qquad
V^2=\mathcal V_r(\Sigma^{-1},(-\infty,0)),
\]
and $V=\mathcal V_{ur}$. We show
\begin{equation}
E\|V^{12}\|^2
\leq
\max\{E\|V^1\|^2,E\|V^2\|^2\}.
\label{eq:two_constraint_bound}
\end{equation}

\emph{Claim 1: active-set probabilities.}
Partition the sample space according to the active set of the two-binding
projection. Let $B_0$ be the event that neither inequality binds after
projection, $B_1$ and $B_2$ the events that only the first or only the
second inequality binds, and $B_{12}$ the event that both bind. These events
are determined by the two constraint scores $(a_1'V,a_2'V)$ up to
probability-zero boundaries. On $B_0$,
$V^{12}=V$. By Lemma~\ref{lem:subset_constraints}, on $B_1$ the
two-binding and first-only projections coincide, and
\[
B_1=\{a_1'V\leq0,\ a_2'V^1>0\}
\]
up to probability-zero boundaries. Conditional on $a_1'V=x\leq0$, $V^1$ is a
zero-mean normal vector with covariance independent of $x$; this follows
from Corollary~\ref{cor:cond_dist_proj}. The remaining condition
$a_2'V^1>0$ is therefore a centered normal halfspace conditional on
$a_1'V=x$. Since $\|V^1\|^2$ is invariant to the sign change
$V^1\mapsto -V^1$, symmetry implies
\[
E[\|V^{12}\|^2\1(B_1)]
=
\frac12
E[\|V^1\|^2\1(a_1'V\leq0)].
\]
Analogously,
\[
E[\|V^{12}\|^2\1(B_2)]
=
\frac12
E[\|V^2\|^2\1(a_2'V\leq0)].
\]
Define
\[
M_1=E[\|V^1\|^2\1(a_1'V\leq0)],
\qquad
M_2=E[\|V^2\|^2\1(a_2'V\leq0)].
\]
Then
\begin{equation}
E\|V^{12}\|^2
=
E[\|V\|^2\1(B_0)]
+
E[\|V^{12}\|^2\1(B_{12})]
+
\frac12(M_1+M_2).
\label{eq:two_partition}
\end{equation}

Suppose first that $M_1\geq M_2$. It is enough to compare the
$B_{12}$ contribution with a region that appears in the risk of $V^1$.
Let
\[
\bar B_1=\{a_1'V>0,\ a_2'V\leq0\}.
\]
We next compute the probabilities of the one-binding active-set events.
The event $B_1$ can be written, up to probability-zero boundaries, as
\[
B_1=\{a_1'V\leq0,\ a_2'V^1>0\}.
\]
Fix $x\leq0$. Conditional on $a_1'V=x$, the first-only projection imposes
the equality $a_1'p=0$, and Corollary~\ref{cor:cond_dist_proj} implies
that $V^1$ is centered normal with covariance matrix not depending on $x$.
Because the two restrictions are nonredundant, $a_2'V^1$ has positive
conditional variance. Therefore
\[
P(a_2'V^1>0\mid a_1'V=x)=\frac12
\]
for almost every $x\leq0$. Since $a_1'V$ is a centered nondegenerate normal
random variable,
\[
P(B_1)
=
E\!\left[
\1\{a_1'V\leq0\}P(a_2'V^1>0\mid a_1'V)
\right]
=
\frac12P(a_1'V\leq0)
=
\frac14.
\]
The same argument with the roles of the two inequalities reversed gives
$P(B_2)=1/4$.

\emph{Claim 2: equality of comparison-region probabilities.}
Since the four active-set events partition the sample space,
\[
P(B_{12})=\frac12-P(B_0).
\]
Moreover,
\[
P(\bar B_1)
=
P(a_1'V>0)-P(B_0)
=
\frac12-P(B_0),
\]
so $P(B_{12})=P(\bar B_1)$.

\emph{Claim 3: conditional second moments.}
Let $A_2$ denote the $2\times k$ matrix with rows $a_1'$ and $a_2'$, and
set
\[
K
=
\Sigma-\Sigma A_2'(A_2\Sigma A_2')^{-1}A_2\Sigma.
\]
Conditional on $A_2V=x$, the event $B_{12}$ is nonrandom. On this event,
inverse-covariance projection imposes both inequalities as equalities, so
Lemma~\ref{lem:general_proj} and Lemma~\ref{lem:cond_normal} imply that
$V^{12}$ has conditional mean zero and covariance $K$. Hence
\[
E[\|V^{12}\|^2\1(B_{12})]
=
\operatorname{tr}(K)P(B_{12}).
\]
On $\bar B_1$, the first-only projection does nothing, so $V^1=V$.
Conditional on $A_2V=x$,
\[
E[\|V\|^2\mid A_2V=x]
=
\|E[V\mid A_2V=x]\|^2+\operatorname{tr}(K)
\geq \operatorname{tr}(K).
\]
Since $\bar B_1$ is also determined by $A_2V$,
\[
E[\|V^1\|^2\1(\bar B_1)]
\geq
\operatorname{tr}(K)P(\bar B_1)
=
E[\|V^{12}\|^2\1(B_{12})].
\]
Using $M_1\geq M_2$ in \eqref{eq:two_partition} therefore gives
\begin{align*}
E\|V^{12}\|^2
&\leq
E[\|V\|^2\1(B_0)]
+
E[\|V^1\|^2\1(\bar B_1)]
+
M_1\\
&=
E\|V^1\|^2.
\end{align*}
If $M_2\geq M_1$, the symmetric argument yields
$E\|V^{12}\|^2\leq E\|V^2\|^2$. This proves
\eqref{eq:two_constraint_bound}.

By Theorem~\ref{thm:one_const}, inverse covariance is pointwise optimal in
every attainable one-binding configuration. Hence the worst risk over $\bar C_1$ is a
lower bound on the minimax value over $\bar C_2$, while
\eqref{eq:two_constraint_bound} shows that every two-binding risk under
inverse covariance is no larger than that lower bound. Therefore
$\Sigma^{-1}$ is minimax over $\bar C_2$.

For Part~2, let $L=\{j:c_j>-\infty\}$ and omit the asymptotically irrelevant
inequalities. The case $|L|\leq1$ follows from Part~2 of Theorem~\ref{thm:one_const}, so
suppose $|L|=2$. If the two finite restriction normals are linearly independent,
let $A_L$ collect their rows and choose $u_0$ with $A_Lu_0=c_L$. Define
\[
X=\Sigma^{-1/2}(\mathcal V_{ur}-u_0),\qquad
\theta=-\Sigma^{-1/2}u_0,\qquad
B=A_L\Sigma^{1/2}.
\]
Then $X\sim N(\theta,I_k)$ and $B\theta=-c_L\geq0$. Moreover,
inverse-covariance projection of $\mathcal V_{ur}$ onto
$\{u:A_Lu\geq c_L\}$ is equivalent to Euclidean projection of $X$ onto
$\{z:Bz\geq0\}$, the restricted Gaussian maximum likelihood estimator.
The two-restriction result of \citet{RuedaSalvador1995} gives, for every
$q\in\mathbb R^k$,
\[
E\!\left[\{q'(X^*-\theta)\}^2\right]
\leq
E\!\left[\{q'(X-\theta)\}^2\right],
\]
where $X^*$ denotes the restricted estimator. Hence
$E[(X^*-\theta)(X^*-\theta)']\preceq I_k$. Since
\[
\mathcal V_r(\Sigma^{-1},c)=\Sigma^{1/2}(X^*-\theta),
\]
taking the trace against $\Sigma^{1/2}H\Sigma^{1/2}\succeq0$ yields
$R_H(\Sigma^{-1},c)\leq\operatorname{tr}(H\Sigma)=R_{ur}$.

It remains to consider linearly dependent normals. If they point in the same
direction, the two inequalities reduce to a single halfspace and
Theorem~\ref{thm:one_const} applies. If $a_2=\lambda a_1$ with $\lambda<0$,
the two inequalities are equivalent to
\[
c_1\leq a_1'u\leq c_2/\lambda,
\]
and this interval contains zero because $c_1,c_2\leq0$. Let
$Z=a_1'\mathcal V_{ur}$ and
$d^*=\Sigma a_1/(a_1'\Sigma a_1)$. The Gaussian decomposition
$\mathcal V_{ur}=d^*Z+\varepsilon$ has $\varepsilon$ centered and independent
of $Z$, while inverse-covariance projection replaces $Z$ by its scalar
projection $\Pi(Z)$ onto the interval. Since $0$ belongs to the interval,
$|\Pi(Z)|\leq|Z|$ draw by draw. The cross term with $\varepsilon$ has mean
zero, so for every $H\succeq0$ projected risk is no larger than unrestricted
risk. This proves Part~2.
\end{proof}

\begin{proof}[Proof of Theorem \ref{thm:suff_cond}]
Let
\[
M_1=\max_{c\in\bar C_1}R_H(\Sigma^{-1},c).
\]
By Theorem~\ref{thm:one_const}, inverse covariance is pointwise optimal at every attainable $c\in\bar C_1$. Therefore, for every positive-definite $\Omega$,
\[
\max_{c\in\bar C_d}R_H(\Omega,c)
\geq
\max_{c\in\bar C_1}R_H(\Omega,c)
\geq M_1.
\]
Under the sufficient condition in the theorem,
\[
\max_{c\in\bar C_d}R_H(\Sigma^{-1},c)=M_1.
\]
Thus $M_1$ is both a lower bound on the boundary-minimax value and the worst boundary risk attained by $\Sigma^{-1}$. Hence $\Sigma^{-1}$ is boundary minimax.
 \end{proof}

\begin{proof}[Proof of Theorem \ref{thm:plugin_boundary_stability}]
Fix $c\in C$ and write
\[
x=\mathcal V_r(\Omega,c),
\qquad
\hat x=\hat{\mathcal V}_r(\hat\Sigma,\hat\Omega,c),
\qquad
D=\hat x-x.
\]
Let $W=\Omega$, $\hat W=\hat\Omega$, $y=\Sigma^{1/2}Z$, and $\hat y=\hat\Sigma^{1/2}Z$. Work on the event that $\hat\Sigma$ and $\hat W$ are positive definite. The variational inequalities for the two metric projections imply
\[
D'\hat W(\hat x-\hat y)\leq0,
\qquad
D'W(x-y)\geq0.
\]
Using these inequalities,
\begin{align*}
D'\hat W D
&=D'\hat W(\hat x-\hat y)+D'\hat W(\hat y-x)\\
&\leq
D'\{(W-\hat W)x+\hat W\hat y-Wy\}\\
&=
D'\{(W-\hat W)x+(\hat W-W)\hat y+W(\hat y-y)\}.
\end{align*}
Let $\lambda=\lambda_{\min}(W)$ and $\hat\lambda=\lambda_{\min}(\hat W)$. Cauchy--Schwarz therefore gives
\[
\|D\|
\leq
\frac{1}{\hat\lambda}
\left(
\|W-\hat W\|(\|x\|+\|\hat y\|)
+
\|W\|\,\|\hat y-y\|
\right).
\]
Because $0\in U(c)$ for every $c\in C$, the projection inequality gives $\|x\|_W\leq\|y\|_W$ and $\|\hat x\|_{\hat W}\leq\|\hat y\|_{\hat W}$. Hence
\[
\|x\|\leq
\sqrt{\frac{\lambda_{\max}(W)}{\lambda}}\,\|y\|,
\qquad
\|\hat x\|\leq
\sqrt{\frac{\lambda_{\max}(\hat W)}{\hat\lambda}}\,\|\hat y\|.
\]
Since $\hat\Sigma\xrightarrow{p}\Sigma\succ0$ and $\hat W\xrightarrow{p}W\succ0$, there are constants $0<\underline m<\bar m<\infty$ such that the event
\[
E_n=\{\lambda_{\min}(\hat\Sigma)\geq\underline m,\ 
\lambda_{\max}(\hat\Sigma)\leq\bar m,\ 
\lambda_{\min}(\hat W)\geq\underline m,\ 
\lambda_{\max}(\hat W)\leq\bar m\}
\]
has probability approaching one. On $E_n$, the preceding inequalities imply that there is a constant $C$, not depending on $c$ or $n$, such that
\[
\|D\|
\leq
C\|Z\|
\left(
\|\hat W-W\|
+
\|\hat\Sigma^{1/2}-\Sigma^{1/2}\|
\right)
\]
and
\[
\|x\|+\|\hat x\|\leq C\|Z\|.
\]
Now
\begin{align*}
\left|\|\hat x\|^2-\|x\|^2\right|
&\leq \|D\|(\|\hat x\|+\|x\|)\\
&\leq
C\|Z\|^2
\left(
\|\hat W-W\|
+
\|\hat\Sigma^{1/2}-\Sigma^{1/2}\|
\right).
\end{align*}
Taking conditional expectations on $E_n$ and using $E\|Z\|^2=k$ yields a bound, conditional on $(\hat\Sigma,\hat W)$, that is $o_p(1)$. The bound does not depend on $c$, so it holds uniformly over $C$. Consistency of $\hat\Sigma$ and $\hat W$, together with continuity of the symmetric square root on matrices with eigenvalues bounded away from zero, imply the desired convergence. The complement of $E_n$ and the zero convention outside the positive-definite event are immaterial because both events have probability approaching zero.
\end{proof}

The same proof yields the path-uniform version used in Section~\ref{subsec:cov_estimation}. In the same normalized coordinates, let
\begin{align*}
\bar K_\Sigma&=\frac{\Sigma}{\operatorname{tr}(\Sigma)},
&
\widehat{\bar K}_\Sigma&=\frac{\hat\Sigma}{\operatorname{tr}(\hat\Sigma)},\\
K_s&=(1-s)\bar K_H+s\bar K_\Sigma,
&
\hat K_s&=(1-s)\bar K_H+s\widehat{\bar K}_\Sigma,
\qquad
\hat\Omega_s=\hat K_s^{-1}.
\end{align*}
Since $\hat\Sigma\xrightarrow{p}\Sigma\succ0$, $\sup_{s\in[0,1]}\|\hat K_s-K_s\|=o_p(1)$. The population path is uniformly positive definite, and the feasible path is uniformly positive definite with probability approaching one. Uniform continuity of matrix inversion on sets with eigenvalues bounded away from zero gives
\[
\sup_{s\in[0,1]}\|\hat\Omega_s-\Omega_s\|=o_p(1).
\]
Replacing $W$ and $\hat W$ in the proof above by $\Omega_s$ and $\hat\Omega_s$ therefore makes the same random constant uniform in $s$, and gives
\[
\sup_{s\in[0,1]}\sup_{c\in C}
\left|\hat R_n(s,c)-R_{\Sigma,\Omega_s}(c)\right|
=o_p(1),
\]
where $\hat R_n(s,c)$ denotes the plug-in risk computed with $\hat\Sigma$ and $\hat\Omega_s$.

\begin{proof}[Proof of Corollary \ref{cor:plugin_boundary_stability}]
Since $\hat\Sigma\xrightarrow{p}\Sigma\succ0$, continuity of matrix inversion on the positive-definite cone gives $\hat\Sigma^{-1}\xrightarrow{p}\Sigma^{-1}$ on an event whose probability approaches one. Thus Theorem~\ref{thm:plugin_boundary_stability} applies with $\Omega=\Sigma^{-1}$ and $\hat\Omega=\hat\Sigma^{-1}$.
Because $\bar C_d$ is finite, Theorem~\ref{thm:plugin_boundary_stability} implies
\[
\max_{c\in\bar C_d}|\hat R_n(c)-R_{\Sigma,\Sigma^{-1}}(c)|=o_p(1).
\]
For any finite set $S$,
\[
\left|\max_{c\in S}\hat R_n(c)-\max_{c\in S}R_{\Sigma,\Sigma^{-1}}(c)\right|
\leq \max_{c\in S}|\hat R_n(c)-R_{\Sigma,\Sigma^{-1}}(c)|.
\]
Applying this inequality to $S=\bar C_d$ and $S=\bar C_1$ yields
\[
|\hat\Delta_n-\Delta|
\leq
2\max_{c\in\bar C_d}|\hat R_n(c)-R_{\Sigma,\Sigma^{-1}}(c)|
=o_p(1).
\]
Thus $\hat\Delta_n\xrightarrow{p}\Delta$. In particular, if $\Delta>0$, the strict failure of the sufficient condition is detected with probability approaching one.
\end{proof}

\section{Additional simulations and computational details}
\label{appendix:sim_additional}

\subsection{Three-constraint counterexample details}
\label{appendix:counterexample_details}

This subsection reports the complete set of boundary risks for the numerical counterexample in Section~\ref{subsec:three_counterexample}. Here, $\mathcal V_{ur}\sim N(0,I_3)$, so the inverse-covariance metric is $I_3$. The three inequality normals are
\[
a_1=(\tau,0,\rho)',\qquad
a_2=(-\tau/2,\sqrt{3}\tau/2,\rho)',\qquad
a_3=(-\tau/2,-\sqrt{3}\tau/2,\rho)',
\]
where $\rho=0.9$ and $\tau=\sqrt{0.19}$. The positive-definite loss used in the risk calculation is
\[
H_\varepsilon=\operatorname{diag}(10^{-4},10^{-4},1),
\]
and the alternative projection metric is $\operatorname{diag}(1,1,1.35)$.

For a single halfspace with normal $a$, define
\[
q_\Omega(a)=\frac{\Omega^{-1}a}{a'\Omega^{-1}a}.
\]
Symmetry of the standard Gaussian distribution implies the closed-form single-halfspace risk
\[
R(a;\Omega,H)=\operatorname{tr}(H)-a'Hq_\Omega(a)
+\frac12(a'a)q_\Omega(a)'Hq_\Omega(a).
\]
This expression verifies the single-halfspace comparison analytically. The risks below are ordinary Monte Carlo estimates based on 300,000 iid Gaussian draws. Parentheses report standard errors. The same draws are used for both projection metrics, so the standard error in the difference column is computed from the paired draw-by-draw risk differences.

\begin{table}[!tbp]
\centering
\caption{Complete boundary risks in the three-constraint counterexample}
\label{tab:counterexample_full}
\footnotesize
\begin{tabular}{lccc}
\toprule
Binding set $S$ & \makecell{Inverse covariance\\$R_S(I_3)$} & \makecell{Alternative metric\\$R_S(\operatorname{diag}(1,1,1.35))$} & Difference\\
\midrule
$\{1\}$ & \makecell{0.592877\\(0.001995)} & \makecell{0.594519\\(0.001995)} & \makecell{-0.001642\\(0.000064)}\\
$\{2\}$ & \makecell{0.592565\\(0.001995)} & \makecell{0.594183\\(0.001995)} & \makecell{-0.001618\\(0.000063)}\\
$\{3\}$ & \makecell{0.592656\\(0.001996)} & \makecell{0.594258\\(0.001995)} & \makecell{-0.001601\\(0.000064)}\\
$\{1,2\}$ & \makecell{0.581040\\(0.002000)} & \makecell{0.569528\\(0.002003)} & \makecell{0.011512\\(0.000055)}\\
$\{1,3\}$ & \makecell{0.581473\\(0.002000)} & \makecell{0.569914\\(0.002003)} & \makecell{0.011559\\(0.000055)}\\
$\{2,3\}$ & \makecell{0.581129\\(0.002000)} & \makecell{0.569577\\(0.002003)} & \makecell{0.011552\\(0.000055)}\\
$\{1,2,3\}$ & \makecell{0.612548\\(0.001987)} & \makecell{0.593276\\(0.001994)} & \makecell{0.019272\\(0.000059)}\\
\bottomrule
\end{tabular}
\end{table}
The identity metric is optimal for each one-binding problem, but the three-binding risk becomes its worst boundary risk. The alternative metric accepts a small increase in each one-binding risk in exchange for a substantially lower risk when several inequalities bind. This lowers the maximum over boundary configurations and establishes the failure of inverse-covariance boundary minimaxity in this example.

\subsection{Computational details}
\label{appendix:replication_details}

The selector and application calculations use scrambled Sobol Gaussian integration, the conservative screen $\kappa=5$, and finite slack box $M=8$.

For modest retained sets, the finite-box diagnostic searches every nonempty retained subset and includes the all-interior value one. The illustrative and ordered-treatment results use the common grid $\mathcal S=\{0,0.05,\ldots,1\}$ and the numerical feasible set
\[
\widehat{\mathcal S}_{\mathrm{num}}
=\{0\}\cup\{s\in\mathcal S\setminus\{0\}:\hat G_8(s)\leq1.004\}.
\]
We minimize $\hat B(s)$ over this set separately in each replication, breaking ties in favor of larger $s$, and then average the selected criteria and losses. Loss projection is admitted analytically and its diagnostic is set to one, so the set is never empty. The $0.004$ allowance accommodates an observed integration offset. In all 831 ordered-treatment replications where the inverse-covariance diagnostic exceeds $1.001$, it equals the unadjusted diagnostic for the analytically safe loss projection. The largest such value is $1.003055$. The allowance is a numerical convention, not a confidence bound for population no-harm. All other displayed diagnostic values are direct numerical estimates. The application separately reports its paired subset-maximized diagnostic. Table~\ref{tab:three_constraint_counterexample} uses 
300,000 iid Gaussian draws. The main computational settings are:
\begin{center}
\footnotesize
\begin{tabular}{p{0.38\linewidth}ccc}
\toprule
Exercise & Risk draws & \makecell{Slack candidates/\\starts/iterations} & Replications\\
\midrule
Illustrative selector & 4,096 & 60 / 2 / 20 & 1,000\\
Ordered-treatment selector & 4,096 & 60 / 2 / 20 & 1,000\\
Gasoline-demand application & 8,192 & 80 / 3 / 30 & one plug-in application\\
\bottomrule
\end{tabular}
\end{center}
The flexible-IV simulations use 1,000 Monte Carlo replications and 10,000 auxiliary draws to evaluate rearrangement loss.

The gasoline-demand application uses the nonredundant low/high-income endpoint diagnostic with the empirical-control loss matrix $H_{\mathrm{all}}$ that averages over the displayed low-, medium-, and high-income levels. Its finite-box diagnostic maximizes over all 1023 nonempty retained subsets. The empirical $H_{\mathrm{all}}$ is treated as the plug-in loss matrix, and Section~\ref{subsec:cov_estimation} does not separately analyze its estimation.

\subsection{Implementation details for the boundary-risk selector}
\label{appendix:selector_details}

The feasible selector estimates the root-$n$ covariance matrix, conservatively screens restrictions, and standardizes retained local slacks by their estimated standard deviations. The population no-harm criterion uses $[0,\infty]^m$. Numerically, finite slacks are searched in $[0,8]^{|L|}$ for each retained subset $L$, while infinite slacks are represented by omitted inequalities. The all-interior case contributes one. Thus $M=8$ is a numerical approximation, not part of the population criterion. Other values of $M$ can be checked if the maximizer is close to the finite-box boundary.

For fixed $s$ and standardized slack vector $\mu$, the projection problem has the dual representation
\[
\lambda_s(z,\mu)=\argmin_{\lambda\geq0}
\left\{\frac12\lambda'Q_s\lambda+\lambda'(z+D\mu)\right\},
\qquad Q_s=A K_s A',
\]
where $K_s=\Omega_s^{-1}$ and $z=A\mathcal V_{ur}$. The path uses the normalized endpoint matrices
\[
K_s=(1-s)\bar K_H+s\bar K_\Sigma,\qquad
\bar K_H=\frac{H^{-1}}{\operatorname{tr}(H^{-1})},
\qquad
\bar K_\Sigma=\frac{\Sigma}{\operatorname{tr}(\Sigma)}.
\]
The primal projected error is $\mathcal V_{ur}+K_sA'\lambda_s$. Conditional Gaussian moments integrate out components not determined by the retained constraint score. For $m\leq5$, we enumerate dual active sets. The application uses the same routine at $s=1$ for its 10-restriction system. The finite-slack search uses space-filling candidates and local refinement on each subset, while $\hat B(s)$ enumerates all attainable nonempty exact-boundary subsets. Within each integration, Gaussian draws are held fixed across candidate slacks. Boundary-risk and full-retained-system safety calculations also reuse draws across $s$. Proper-subset safety calculations use path-dependent scrambles. Reported selections and average curves use the common grid in Appendix~\ref{appendix:replication_details}, and plotted losses come from the same saved replications as the tables. The application instead evaluates paired excess losses and adds back exact unrestricted risk.

The illustrative design in Section~\ref{subsec:sim_selector} uses $X_i\sim N(\theta_0,I_{10})$, $n=250$, coordinatewise nonnegativity, and
\[
H=\begin{pmatrix}0.05I_5+\mathbf1\mathbf1'&0\\0&0.05I_5\end{pmatrix}.
\]
The first $m\in\{1,3,5\}$ coordinates are locally relevant. Section~\ref{subsec:sim_selector} considers exact-boundary profiles, and Table~\ref{tab:selector_appendix} also reports mixed profiles with standardized slacks $(0,0.5,1,1.5,2)$ truncated to length $m$. All binding constraints are retained. The screen retains every locally relevant constraint in all exact-boundary replications and in 99.6\% of the five-constraint mixed-profile replications. There are no false positives.

\begin{table}[!htbp]
\centering
\caption{Additional results for the illustrative selector design}
\label{tab:selector_appendix}
\footnotesize
\begin{tabular}{llrrrrr}
\toprule
$m$ & Profile & $\bar s_{BR}$ & $L_0$ & $L_1$ & $L_{BR}$ & $\overline{\hat B(\hat s_{BR})}$\\
\midrule
1 & Boundary & 1.000 & 0.02126 & 0.01983 & 0.01983 & 0.903\\
1 & Mixed & 1.000 & 0.02126 & 0.01983 & 0.01983 & 0.903\\
3 & Boundary & 1.000 & 0.02090 & 0.01947 & 0.01947 & 0.922\\
3 & Mixed & 1.000 & 0.02103 & 0.01803 & 0.01803 & 0.922\\
5 & Boundary & 0.722 & 0.01109 & 0.02426 & 0.02029 & 0.933\\
5 & Mixed & 0.723 & 0.02043 & 0.01756 & 0.01734 & 0.933\\
\bottomrule
\end{tabular}
\begin{minipage}{0.92\linewidth}
\vspace{2mm}
{\footnotesize\begin{singlespace}
Notes: Results use 1,000 replications, estimated covariance matrices, and the conservative screen. $L_0$, $L_1$, and $L_{BR}$ are mean finite-sample $H$-losses for loss projection, inverse covariance, and the replication-specific boundary-risk selector. $\overline{\hat B(\hat s_{BR})}$ averages the normalized criterion at each selected value. Selection uses the common grid, $\hat G_8(s)\leq1.004$ for $s>0$, and the analytical guarantee at $s=0$, as detailed in Appendix~\ref{appendix:replication_details}.
\end{singlespace}}
\end{minipage}
\end{table}

\subsection{Ordered-treatment simulation details}
\label{appendix:ordered_treatment_details}

Section~\ref{subsec:sim_treatment} defines the ordered-treatment-effects estimator, its covariance matrix, and the monotonicity restrictions. This appendix records the local configurations and sampling designs used in the simulation grid. Let $d=0$ denote the common control group and $d=1,\ldots,5$ the ordered treatment intensities. The maintained restrictions are
\[
0\leq\theta_1\leq\theta_2\leq\cdots\leq\theta_5.
\]
Writing $\delta_1=\theta_1$ and $\delta_j=\theta_j-\theta_{j-1}$ for $j\geq2$, the local designs parameterize increments by
\[
\delta_j=\frac{\mu_j\tau_{j,n}}{\sqrt n},
\qquad \tau_{j,n}^2=a_j'\Sigma_n a_j,
\]
whenever the increment is not held fixed away from zero. Table~\ref{tab:ordered_shapes_app} gives the designs.

\begin{table}[!htbp]
\centering
\caption{Ordered-treatment shape configurations}
\label{tab:ordered_shapes_app}
\footnotesize
\begin{tabular}{lll}
\toprule
Shape configuration & Fixed positive increments & Local standardized slacks\\
\midrule
Interior & $(\delta_1,\ldots,\delta_5)=(0.20,0.15,0.12,0.10,0.08)$ & --\\
One binding & $(\delta_1,\ldots,\delta_4)=(0.20,0.15,0.12,0.10)$ & $\mu_5=0$\\
Three binding & $(\delta_1,\delta_2)=(0.20,0.15)$ & $(\mu_3,\mu_4,\mu_5)=(0,0,0)$\\
Three local mixed & $(\delta_1,\delta_2)=(0.20,0.15)$ & $(\mu_3,\mu_4,\mu_5)=(0,0.5,2)$\\
All local mixed & -- & $(\mu_1,\ldots,\mu_5)=(0,0.5,1,1.5,2)$\\
Stress A/B & -- & $(\mu_1,\ldots,\mu_5)=(0,8,8,8,8)$\\
\bottomrule
\end{tabular}
\begin{minipage}{0.92\linewidth}
\vspace{2mm}
{\footnotesize\begin{singlespace}
Notes: Local standardized slacks are reported only for increments placed at the $n^{-1/2}$ scale. Fixed positive increments are asymptotically interior and are not assigned local slack values.
\end{singlespace}}
\end{minipage}
\end{table}

One baseline sampling design is balanced and homoskedastic. Unbalanced homoskedastic sampling uses shares $(0.35,0.10,0.10,0.15,0.15,0.15)$. The unbalanced heteroskedastic design uses the same shares with standard deviations $(1,1.4,1.2,1,0.8,1.5)$. Stress Case A uses shares $(0.05,0.19,0.19,0.19,0.19,0.19)$ and unit standard deviations. Stress Case B uses balanced shares and standard deviations $(3,1,1,1,1,1)$. Sample sizes are $n\in\{100,400,1600\}$. The feasible selector uses 1,000 replications in each cell, estimated covariance matrices, and the conservative screen. Table~\ref{tab:ordered_baseline_selector_app} reports the cell-level results.

\begin{table}[!htbp]
\centering
\caption{Ordered-treatment selector results by design cell}
\label{tab:ordered_baseline_selector_app}
\scriptsize
\begin{tabular}{llrrrrrr}
\toprule
Sampling & Shape configuration & \multicolumn{2}{c}{$n=100$} & \multicolumn{2}{c}{$n=400$} & \multicolumn{2}{c}{$n=1600$}\\
 & & $\overline{\hat B(\hat s_{BR})}$ & $L_{BR}/L_0$ & $\overline{\hat B(\hat s_{BR})}$ & $L_{BR}/L_0$ & $\overline{\hat B(\hat s_{BR})}$ & $L_{BR}/L_0$\\
\midrule
\multicolumn{8}{l}{Panel A: Baseline sampling designs}\\
Balanced homosk. & Interior & 0.957 & 0.930 & 0.954 & 0.952 & 0.952 & 0.997\\
 & One binding & 0.957 & 0.930 & 0.954 & 0.950 & 0.952 & 0.998\\
 & Three binding & 0.957 & 0.976 & 0.955 & 0.955 & 0.952 & 0.995\\
 & Three local mixed & 0.958 & 0.923 & 0.954 & 0.968 & 0.953 & 0.998\\
 & All local mixed & 0.957 & 0.951 & 0.954 & 0.949 & 0.952 & 0.920\\
\addlinespace[2pt]
Unbalanced homosk. & Interior & 0.945 & 0.996 & 0.941 & 0.988 & 0.940 & 0.997\\
 & One binding & 0.945 & 0.976 & 0.941 & 0.985 & 0.940 & 1.000\\
 & Three binding & 0.945 & 0.994 & 0.941 & 0.996 & 0.940 & 1.006\\
 & Three local mixed & 0.945 & 0.984 & 0.942 & 0.992 & 0.940 & 0.999\\
 & All local mixed & 0.944 & 0.978 & 0.941 & 0.977 & 0.940 & 0.985\\
\addlinespace[2pt]
Unbalanced heterosk. & Interior & 0.959 & 0.982 & 0.961 & 0.954 & 0.961 & 0.977\\
 & One binding & 0.960 & 0.988 & 0.960 & 0.974 & 0.961 & 0.986\\
 & Three binding & 0.959 & 1.025 & 0.961 & 1.002 & 0.961 & 1.008\\
 & Three local mixed & 0.960 & 0.957 & 0.961 & 0.967 & 0.961 & 0.982\\
 & All local mixed & 0.960 & 0.988 & 0.961 & 0.979 & 0.961 & 0.971\\
\midrule
\multicolumn{8}{l}{Panel B: Stress sampling designs}\\
Stress: small control & First binding, others far & 0.682 & 0.747 & 0.660 & 0.725 & 0.654 & 0.759\\
Stress: noisy control & First binding, others far & 0.594 & 0.617 & 0.580 & 0.626 & 0.578 & 0.627\\
\bottomrule
\end{tabular}
\begin{minipage}{0.94\linewidth}
\vspace{2mm}
{\footnotesize\begin{singlespace}
Notes: Each cell uses 1,000 replications. $\overline{\hat B(\hat s_{BR})}$ averages normalized boundary risk at the replication-specific selected values. $L_{BR}/L_0$ compares the selector's mean Euclidean loss with loss projection. The plug-in sufficient condition holds throughout. Selection applies the common-grid rule with the $0.004$ allowance in Appendix~\ref{appendix:replication_details} and chooses inverse covariance in every replication. The maximum cell-average endpoint diagnostic exceeds one by about $0.00021$, and all truly binding inequalities are retained.
\end{singlespace}}
\end{minipage}
\end{table}

The feasible covariance matrix replaces group variances by within-group sample variances. In a separate endpoint benchmark implemented in R, we compare the two endpoint projections using the true covariance matrix (oracle) and its estimate (feasible). Across the 45 baseline endpoint-risk comparisons, oracle inverse-covariance projection has lower point-estimated risk in 44 cells and feasible inverse-covariance projection in 40 cells. The mean feasible-to-oracle risk penalty for inverse-covariance projection is 1.83\% at $n=100$, 0.30\% at $n=400$, and 0.06\% at $n=1600$. Thus the conclusion from Table~\ref{tab:ordered_baseline_selector_app} is not driven by using estimated covariance matrices.

Figure~\ref{fig:treatment_selector} shows the path for $n=400$, balanced homoskedastic sampling, and standardized slacks $(0,0.5,1,1.5,2)$. The average finite-box diagnostic is close to one and the boundary criterion favors inverse covariance. Mean finite-sample loss has a shallow interior minimum, with loss at $s=1$ close to that minimum. The difference reflects the distinction between worst-face selection and risk at a particular DGP.

\begin{figure}[!htbp]
\centering
\caption{Boundary-risk selector in the ordered-treatment design}
\label{fig:treatment_selector}
\includegraphics[width=\textwidth]{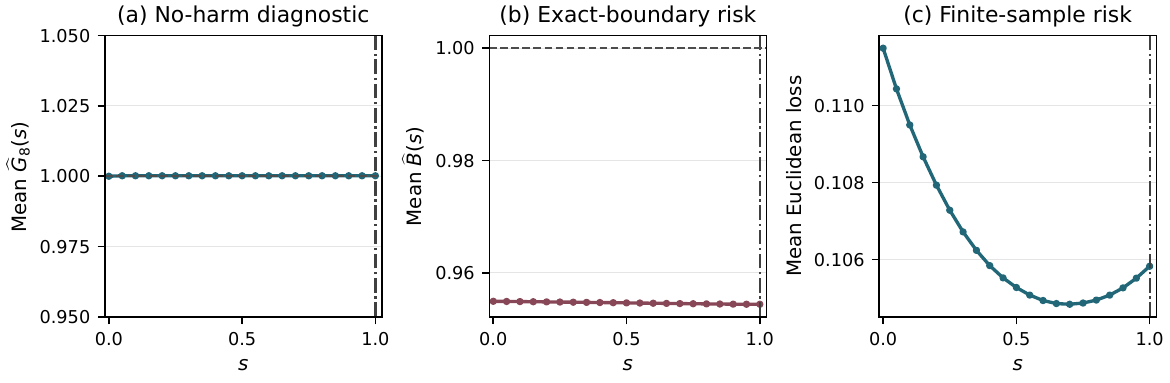}
\begin{minipage}{0.92\linewidth}
\vspace{1mm}
{\footnotesize\begin{singlespace}
Notes: Average finite-box diagnostic $\hat G_8(s)$, boundary-risk criterion $\hat B(s)$, and mean Euclidean loss for $n=400$, balanced homoskedastic sampling, and standardized slacks $(0,0.5,1,1.5,2)$. All panels use the same 1,000 replications as the tables. The vertical marker is $\bar s_{BR}$, and the horizontal dashed line is one. The value $\hat G_8(0)=1$ is imposed analytically.
\end{singlespace}}
\end{minipage}
\end{figure}

\FloatBarrier
\subsection{Flexible-IV simulation details}
\label{appendix:flexible_iv_details}

This subsection gives the data-generating process used in the flexible-IV simulations in Section~\ref{subsec:sim_npiv}. All functions have compact domain $[0,1]$. The structural equation is
\[
Y_i=g(X_i)+\varepsilon_i,
\qquad
E[\varepsilon_i\mid Z_i]=0.
\]
Let $\xi_i,e_i,\nu_i,\kappa_i$ be independent standard normal variables. The instrument and endogenous covariate are
\[
Z_i=\Phi(\xi_i),
\qquad
X_i=\Phi\!\left(\rho \xi_i+\sqrt{1-\rho^2}\,e_i\right),
\]
so that both variables are uniformly distributed on $[0,1]$ and $\rho$ controls instrument strength. The error is
\[
\varepsilon_i
=
\sigma\left(\eta e_i+\sqrt{1-\eta^2}\nu_i\right)
+4\kappa_i 1\{Z_i>0.75\}.
\]
The second term makes the design heteroskedastic while preserving the conditional moment restriction. We set $\rho=0.7$, $\sigma=0.5$, and $\eta=0.3$.

The function $g:[0,1]\rightarrow \mathbb{R}$ is approximated by the basis $(1,x,x^2)$ and the instrument basis is $(1,z,z^2)$. The two designs in Section~\ref{subsec:sim_npiv} set $g(x)=x^2$ and $g(x)=0$, respectively. The monotonicity restriction on $[0,1]$ is imposed through $\theta_2\geq0$ and $\theta_2+2\theta_3\geq0$. Integrated squared error corresponds to the coefficient loss matrix
\[
H=
\begin{pmatrix}
1 & 1/2 & 1/3\\
1/2 & 1/3 & 1/4\\
1/3 & 1/4 & 1/5
\end{pmatrix}.
\]
The reported MISE is $E[(\hat\theta-\theta)'H(\hat\theta-\theta)]$, scaled by 1,000.

\begin{figure}[H]
\centering
\caption{Flexible IV with one binding monotonicity constraint}
\captionsetup[subfigure]{position=top,justification=centering,singlelinecheck=true,skip=1pt}
\begin{subfigure}[t]{0.30\textwidth}
\centering
\caption{2SLS}
{\footnotesize MISE $\times 1000$ = 24.86\par}
\vspace{0.25em}
\includegraphics[width=\linewidth]{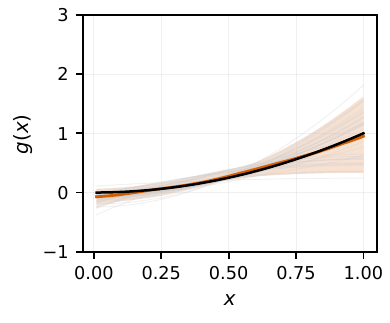}
\end{subfigure}\hfill
\begin{subfigure}[t]{0.30\textwidth}
\centering
\caption{Eff. IV / inv.-cov.}
{\footnotesize MISE $\times 1000$ = 15.69\par}
\vspace{0.25em}
\includegraphics[width=\linewidth]{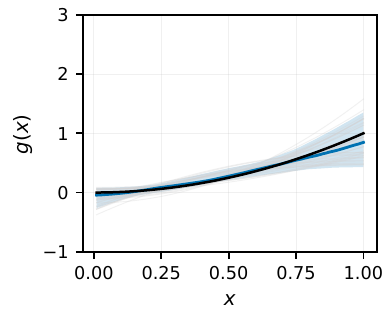}
\end{subfigure}\hfill
\begin{subfigure}[t]{0.30\textwidth}
\centering
\caption{Rearrangement}
{\footnotesize MISE $\times 1000$ = 32.42\par}
\vspace{0.25em}
\includegraphics[width=\linewidth]{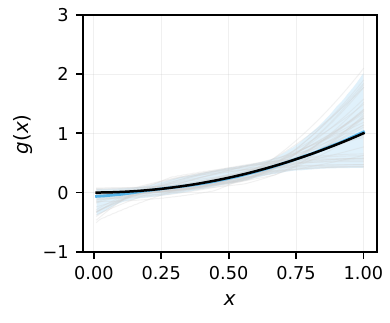}
\end{subfigure}
\vspace{0.8em}

\begin{subfigure}[t]{0.30\textwidth}
\centering
\caption{Euclidean projection}
{\footnotesize MISE $\times 1000$ = 757.91\par}
\vspace{0.25em}
\includegraphics[width=\linewidth]{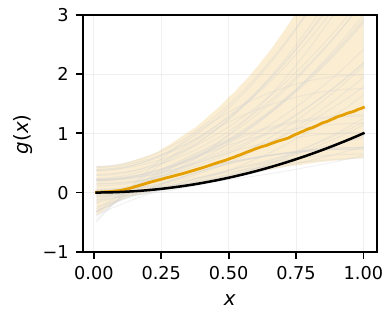}
\end{subfigure}\hfill
\begin{subfigure}[t]{0.30\textwidth}
\centering
\caption{Loss projection}
{\footnotesize MISE $\times 1000$ = 26.69\par}
\vspace{0.25em}
\includegraphics[width=\linewidth]{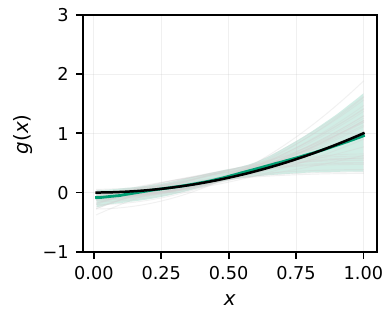}
\end{subfigure}\hfill
\begin{subfigure}[t]{0.30\textwidth}
\centering
\caption{Unrestricted}
{\footnotesize MISE $\times 1000$ = 47.05\par}
\vspace{0.25em}
\includegraphics[width=\linewidth]{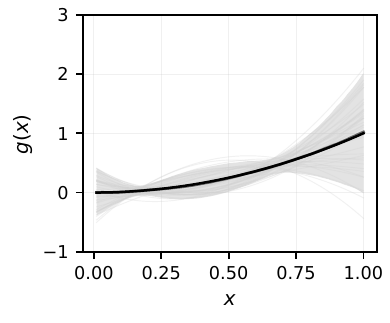}
\end{subfigure}
\begin{minipage}{0.90\linewidth}
\vspace{1mm}
{\footnotesize\begin{singlespace}
Notes: The panels use the common vertical scale $[-1,3]$. Colored lines and shading show pointwise medians and 5th--95th percentiles across all 1,000 replications. Black lines show the true function and faint lines the first 40 replications. The Euclidean projection panel is clipped because its dispersion is much larger. The identity GMM estimator is reported in Table~\ref{tab:npiv_main}; visually, it is close to 2SLS in this design.
\end{singlespace}}
\end{minipage}
\label{fig:npiv_quadratic_panels}
\end{figure}

\begin{figure}[!t]
\centering
\caption{Flexible IV with two binding monotonicity constraints}
\captionsetup[subfigure]{position=top,justification=centering,singlelinecheck=true,skip=1pt}
\begin{subfigure}[t]{0.30\textwidth}
\centering
\caption{2SLS}
{\footnotesize MISE $\times 1000$ = 14.46\par}
\vspace{0.25em}
\includegraphics[width=\linewidth]{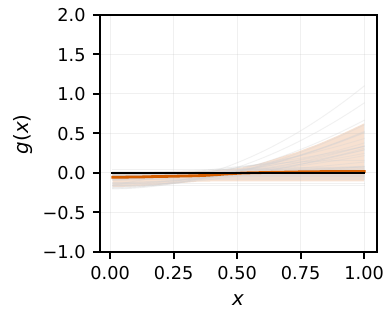}
\end{subfigure}\hfill
\begin{subfigure}[t]{0.30\textwidth}
\centering
\caption{Eff. IV / inv.-cov.}
{\footnotesize MISE $\times 1000$ = 4.09\par}
\vspace{0.25em}
\includegraphics[width=\linewidth]{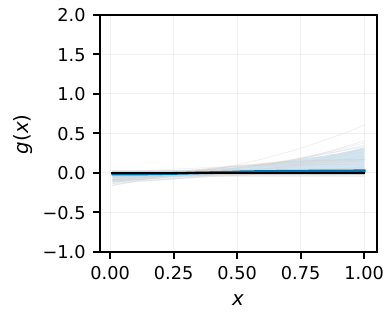}
\end{subfigure}\hfill
\begin{subfigure}[t]{0.30\textwidth}
\centering
\caption{Rearrangement}
{\footnotesize MISE $\times 1000$ = 49.41\par}
\vspace{0.25em}
\includegraphics[width=\linewidth]{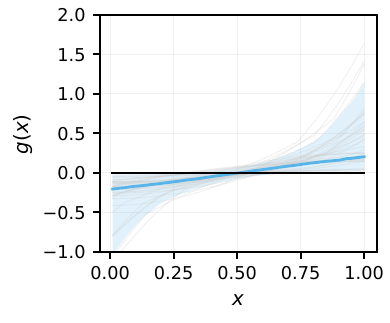}
\end{subfigure}
\vspace{0.8em}

\begin{subfigure}[t]{0.30\textwidth}
\centering
\caption{Euclidean projection}
{\footnotesize MISE $\times 1000$ = 1041.56\par}
\vspace{0.25em}
\includegraphics[width=\linewidth]{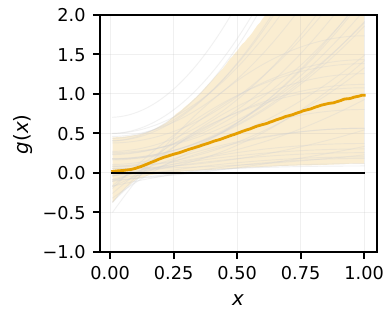}
\end{subfigure}\hfill
\begin{subfigure}[t]{0.30\textwidth}
\centering
\caption{Loss projection}
{\footnotesize MISE $\times 1000$ = 17.75\par}
\vspace{0.25em}
\includegraphics[width=\linewidth]{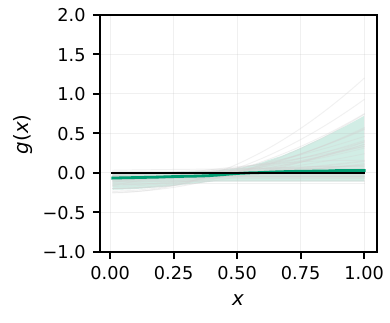}
\end{subfigure}\hfill
\begin{subfigure}[t]{0.30\textwidth}
\centering
\caption{Unrestricted}
{\footnotesize MISE $\times 1000$ = 49.42\par}
\vspace{0.25em}
\includegraphics[width=\linewidth]{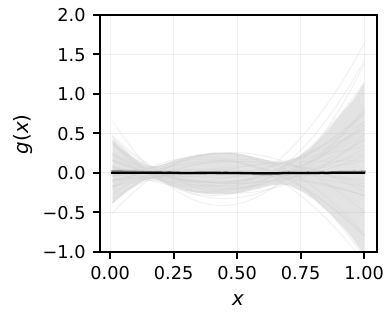}
\end{subfigure}
\begin{minipage}{0.90\linewidth}
\vspace{1mm}
{\footnotesize\begin{singlespace}
Notes: The panels use the common vertical scale $[-1,2]$. Colored lines and shading show pointwise medians and 5th--95th percentiles across all 1,000 replications. Black lines show the true function and faint lines the first 40 replications. The Euclidean projection panel is clipped because its dispersion is much larger. The constant function has two binding inequalities.
\end{singlespace}}
\end{minipage}
\label{fig:npiv_constant_panels}
\end{figure}

\FloatBarrier
\subsection{Illustrative geometry in two dimensions}
\label{appendix:geometry_illustrations}

This subsection complements the main-text illustration in Figure~\ref{fig:linear_iv_projection_geometries}. The purpose is not to add Monte Carlo evidence for Section~\ref{sec:simulations}, but to make the geometry of projection visible. Figure~\ref{fig:linear_iv_projection_geometries} uses a just-identified linear IV design with two regressors and two instruments. The raw instruments are binary with \(P(Z_1=1)=1/2\) and \(P(Z_2=1\mid Z_1)=1/4+Z_1/2\), recentered before estimation. The structural error has standard deviation one when the two raw instruments agree and 0.2 otherwise, the two regressors are \(X_j=0.1\varepsilon+0.9Z_j\), and the true coefficient is \(\theta_0=(0,0)'\). The maintained restriction is \(\theta_2\geq0\). The figure uses \(n=10{,}000\), 1,000 Monte Carlo replications, and compares the projection geometries generated by 2SLS, efficient IV, the identity GMM weight, and \(\Omega=I\). For this illustration, the efficient weight is computed using the simulated structural errors.

Figure~\ref{fig:active_face_geometry} abstracts further from IV and uses a Gaussian projection experiment to show that pointwise rankings of projection geometries can change with the active face of the constraint set.

\begin{figure}[!htbp]
\centering
\caption{Active faces and pointwise projection risk}
\label{fig:active_face_geometry}
\includegraphics[width=0.86\textwidth]{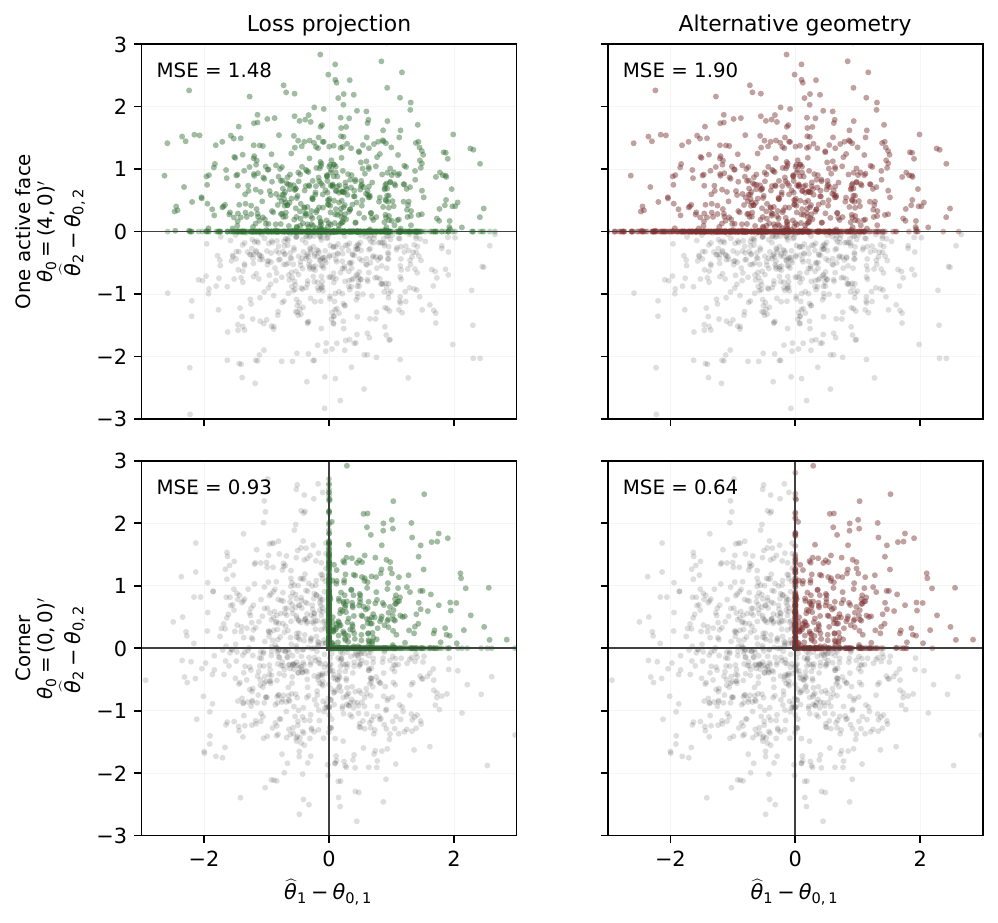}
\begin{minipage}{0.90\linewidth}
\vspace{1mm}
{\footnotesize\begin{singlespace}
Notes: The feasible set is the nonnegative orthant and the panels plot estimation errors relative to the true value. The top row considers a point on one active face, while the bottom row considers the corner. The same two projection geometries are compared in both rows. The example illustrates why a projection geometry that is pointwise attractive for one active set need not remain pointwise attractive for another active set.
\end{singlespace}}
\end{minipage}
\end{figure}

The feasible set is $\mathbb R^2_+$, the loss and sampling covariance are both the identity, and the two geometries are $\Omega=I_2$ and $\Omega=\begin{pmatrix}1&-0.9\\-0.9&1\end{pmatrix}^{-1}$. The two true values are $(4,0)'$ and $(0,0)'$. The figure plots estimation errors relative to these true values, and the reported MSEs use 4,000 Gaussian draws.
\FloatBarrier

\section{Additional application results}
\label{appendix:app_additional}

\begin{table}[H]
\centering
\caption{Boundary risk by number of binding inequalities}
\label{tab:app_test_minimax}
\footnotesize
\renewcommand{\arraystretch}{0.9}
\begin{tabular}{rrr}
\toprule
Binding inequalities & Attainable faces & Maximum relative risk \\
\midrule
1 & 10 & 0.892\\
2 & 45 & 0.866\\
3 & 120 & 0.797\\
4 & 210 & 0.785\\
5 & 252 & 0.770\\
6 & 210 & 0.770\\
7 & 120 & 0.540\\
8 & 45 & 0.280\\
9 & 10 & 0.098\\
10 & 1 & 0.050\\
\bottomrule
\end{tabular}
\begin{minipage}{0.90\linewidth}
\vspace{1mm}
{\footnotesize\begin{singlespace}
Notes: For each number of binding inequalities, the last column reports the largest estimated exact-boundary risk of efficient IV/inverse-covariance projection among faces of that size, normalized by $\operatorname{tr}(H_{\mathrm{all}}\hat\Sigma)$. All 1023 attainable nonempty faces are evaluated, using 8,192 scrambled Sobol draws per face with subset-specific scrambles.
\end{singlespace}}
\end{minipage}
\end{table}

The price and instrument bases use internal knots at the empirical quartiles
and repeated boundary knots at the sample extrema. Each six-function quadratic
spline basis is evaluated on its closed support, including the upper endpoint,
and sums to one throughout. Derivative restrictions use the five distinct price
knots, with one-sided derivatives at the endpoints and no inward displacement.
The covariance estimate is the heteroskedasticity-robust sandwich covariance of
the just-identified IV estimator, whose inverse gives the efficient-IV/inverse-covariance geometry.

The empirical loss matrix $H_{\mathrm{all}}$ is positive definite in the
16-dimensional parameterization, with smallest eigenvalue $1.31\times10^{-5}$
and condition number $3.87\times10^6$, so loss projection is well defined.
The low/high-income restriction matrix has rank 10, making all 1023 nonempty
subsets attainable. As explained in Section~\ref{sec:application}, this system
also enforces monotonicity at every intervening income.

\end{document}